\documentclass[journal,final]{IEEEtran}

\usepackage{amsmath,amsfonts,amssymb,paralist,latexsym,epsfig,color,rotating,times,float,subfigure,algorithm,verbatim,enumitem,amsthm,cancel}
\usepackage{tikz}

\newcommand{\lag}{\Delta}
\newcommand{\lagc}{l}
\newcommand{\lSC}{{\text{SC}}}
\newcommand{\mse}{\operatorname{MSE}}

\makeatletter
\def\nl#1#2{\begingroup
    #2%
    \def\@currentlabel{#2}%
    \phantomsection\label{#1}\endgroup
}
\makeatother
{\ensuremath{
\begin{empheq}[box=\fbox]{align}
{#1}
\end{empheq}
}}

\newtheorem{theoremi}{Theorem}

\newtheorem{theorem}            {Theorem}

\newtheorem{corollary}          [theorem]{Corollary}
\newtheorem{proposition}        [theorem]{Proposition}

\newtheorem{lemma}              [theorem]{Lemma}

\newcommand{\cm}{m}
\newcommand{\cmi}{\cm_1}
\newcommand{\cmii}{\cm_2}
{
\begin{mdframed}
\par\noindent\textbf{#1:}\begin{rmfamily}\noindent}% 
{\end{rmfamily}
\end{mdframed}
}

\newcommand{\matd}{A}
\newcommand{\mato}{L}

\newcommand{\matii}{H}

\newcommand{\matsii}{{H}}

\newcommand{\prevstate}{\bar{\state}}
\newcommand{\by}{\bar{y}}
\newcommand{\bz}{\bar{z}}

\newcommand{\sgn}{\operatorname{sgn}}

\newcommand{\bcdf}{\bar{\cdf}}

\newcommand{\pdf}{p}
\newcommand{\cdf}{F}
\newcommand{\prob}{\mathbb{P}}
\newcommand{\E}                 {\Bbb{E}}

\newcommand{\bstate}{\bar{\state}}

\newcommand{\z}{z}

\newcommand{\borelset}{S}

\newcommand{\obs}{y}
\newcommand{\obsa}{y^{(\action)}}

\newcommand{\Obs}{Y}
\newcommand{\Obsi}{\Obs^{(1)}}
\newcommand{\Obsii}{\Obs^{(2)}}
\newcommand{\Obsa}{\Obs^{(\action)}}  
\newcommand{\Snoisea}{\Snoise^{(\action)}}

\newcommand{\snoise}{w}
\newcommand{\Snoise}{W}

\newcommand{\level}{g}

\newcommand{\state}{x}
\newcommand{\State}{X}
\newcommand{\statespace}{\Bbb{X}}
\newcommand{\obspace}{\Bbb{Y}}
\newcommand{\statedim}{\mathcal{X}}
\newcommand{\obsdim}{\mathcal{Y}}

\newcommand{\fun}{\phi}

\newcommand{\oprob}{B}
\newcommand{\tp}{P}

\newcommand{\belief}{\pi}

\newcommand{\Belief}{\Pi(\statedim)}
\newcommand{\Belieft}{\Pi(2)}

\newcommand{\ole}{\stackrel{\text{defn}}{=}}

\newcommand{\gr}{\geq_r}
\newcommand{\lr}{\leq_r}
\newcommand{\gs}{\geq_s}

\newcommand{\filterd}{\sigma}

\newcommand{\filter}{T}

\newcommand{\reals}{{\rm I\hspace{-.07cm}R}}

\newcommand{\beq}{\begin{equation}}
\newcommand{\eeq}{\end{equation}}

\newcommand{\p}{\prime}

\newcommand{\one}{\mathbf{1}}
\newcommand{\ones}{\mathbf{1}}

\newcommand{\diag}{\textnormal{diag}}

\newcommand{\reward}{r}

\newcommand{\action}{u}

\newcommand{\actionspace}{\,\mathcal{U}}
\newcommand{\actiondim}{U}

\newcommand{\discount}{\rho}

 \newcommand{\Ep}{\E_{\policy}}

\newcommand{\policy}{\mu}

\newcommand{\optpolicy}{\policy^*}

\newcommand{\valuef}{V}

\renewcommand{\time}{k}

\newcommand {\policyl} {\underline{\mu}}

\newcommand{\bj}{l}

\newcommand{\argmin}{\operatornamewithlimits{argmin}}
\newcommand{\argmax}{\operatornamewithlimits{argmax}}

\newcommand{\barray}{\begin{array}{ll}}
\newcommand{\earray}{\end{array}}

\newcommand{\lami}{\lambda_a}
\newcommand{\lamii}{\lambda_b}
\newcommand{\lamiii}{\lambda_c}
\newcommand{\lamiv}{\lambda_d}
\newcommand{\obsl}{a}
\newcommand{\obsu}{b}

\newcommand{\horizon}{N}
\newcommand{\rstop}{\reward_{s}}

\makeatletter
\newenvironment{subdef}[1]{%
  \def\subdefcounter{#1}%
  \refstepcounter{#1}%
  \protected@edef\theparentnumber{\csname the#1\endcsname}%
  \setcounter{parentnumber}{\value{#1}}%
  \setcounter{#1}{0}%
  \expandafter\def\csname the#1\endcsname{\theparentnumber.\Alph{#1}}%
  \ignorespaces
}{%
  \setcounter{\subdefcounter}{\value{parentnumber}}%
  \ignorespacesafterend
}
\makeatother
\newcounter{parentnumber}

\newtheorem{defn}{Definition}

\begin{document}

\title{Convex Stochastic Dominance in Bayesian Localization, Filtering   and Controlled Sensing   POMDPs}

\author{Vikram Krishnamurthy\thanks{This research was funded in part by the U.S. Army Research Office under grant W911NF-19-1-0365, U.S. Air Force Office of Scientific Research under grant FA9550-18-1-0007 and National Science Foundation under grant 1714180.}, {\em Fellow IEEE}\\
 Electrical \& Computer
Engineering,  \\ Cornell University, USA. \\
vikramk@cornell.edu}

\maketitle

\begin{abstract}
  This paper provides conditions on the observation probability distribution in  Bayesian localization and optimal filtering   so that the conditional mean estimate satisfies convex stochastic dominance. Convex dominance  allows us to compare
 the unconditional mean square error between two optimal Bayesian state  estimators over arbitrary time horizons instead of using brute force Monte-Carlo computations. The  proof uses two key ideas from microeconomics, namely, integral precision dominance and aggregation of single crossing. 
  The convex dominance result is then used to give sufficient conditions  so that the optimal policy of a controlled sensing two-state
  partially observed Markov decision process (POMDP) is  lower bounded by a myopic policy. 
  Numerical examples  are presented where the Shannon capacity of the observation distribution using one sensor dominates that of another, and convex dominance holds but  Blackwell dominance does not hold. These illustrate the usefulness of the main result in localization, filtering and controlled sensing applications.
  \end{abstract}


  \section{Introduction} \label{sec:intro}
  Consider the following Bayesian localization problem:
  an underlying random variable    $\State\in \reals$ with prior $\belief_0$ is observed
  via the discrete time noisy observation process $\{\Obs_k\}$ where each observation
  $\Obs_k$  has conditional cumulative distribution function (cdf)
  $ \cdf(\obs| \state)$. (We use upper case for random variables and lower case for realizations.)
  Bayesian localization  is concerned with recursively computing the posterior
  distribution  $\belief_k = \pdf(\state| \obs_{1:k})$, $k=1,2,\ldots$  of the 
  state $\state$ given observation sample path sequence $\obs_{1:k}= (y_1,\ldots,y_k)$ and prior $\belief_0$.
  The posterior distribution $\belief_k$  is then used to compute the conditional mean 
  estimate of the state $\State$ given $k$ observations as  $$\cm(\obs_{1:k},\belief_0) = \int_\reals \state \belief_k(x) dx$$
  where we have indicated the explicit dependence on the prior $\belief_0$.

Let $\Obs_{1:k}$ denote the sequence of random variables $(\Obs_1,\ldots,\Obs_k)$.
A natural question is: how accurate is the conditional mean state estimate
$ \cm(\Obs_{1:k},\belief_0)$?
Clearly  $\cm(\Obs_{1:k},\belief_0)$  is the minimum mean square error estimate (more generally it minimizes a Bregman loss), i.e.,
for all priors $\belief_0$,
$$\mse\{\cm(\Obs_{1:k},\belief_0)\}= \argmin_{g} \E\{\big(\State - g(\Obs_{1:k},\belief_0)\big)^2\} $$
over the class of all Borel  functions $g$.
But unfortunately,
apart from the well known linear Gaussian  case,\footnote{In the linear Gaussian case, the MSE is computed by  the
  Kalman filter  covariance update (Riccati equation) which is completely determined by  the model parameters.}
$\mse\{\cm(\Obs_{1:k},\belief_0)\}$ 
can only be estimated\footnote{Computing the conditional
  $\mse\{\cm(\obs_{1:k}),\belief_0 \}$ based on a specific observation sample path $\obs_{1:k}$ given the posterior $\belief_k$ is straightforward but not useful since it only holds for the specific sample path $\obs_{1:k}$. We are interested in characterizing its expectation, i.e., the unconditional MSE.}  via Monte-Carlo simulations.  Thus for general state space models, there is strong motivation to derive  analytical results that
compare $\mse\{\cm(\Obs_{1:k},\belief_0)\}$ for different observation models. Specifically consider  two sensors, %with observation processes
% \footnote{Strictly speaking we should use the notation $\obs_k^\action$ since the sequence of random variables are generated from the conditional distribution $\cdf_u(y|x)$. But for notational convenience we use $\obs_k$.}
%$\{\Obsi_k\}$,  $\{\Obsii_k\}$, respectively,   
where sensor 1 records observation random variables $\Obsi_k $ of state $\State$ with conditional distribution $ \cdf_1(\obs| \state)$, $k=1,2,\ldots$  and  sensor 2 records observation random variables $ \Obsii_k$ of $\State$ with 
conditional distribution $\cdf_2(\obs| \state)$, $k=1,2,\ldots$. Then
which sensor yields a smaller MSE  for the conditional mean estimate at time $k$?

In this paper,  we give sufficient conditions on the observation probabilities so that the conditional mean estimate $\cm_1(\Obsi_{1:k},\belief_0)$ of sensor 1 is   convex stochastic dominated by the estimate $\cm_2(\Obsii_{1:k},\belief_0)$ of sensor 2.
Informally,  our main result is:  
\begin{theoremi} (Informal)  Consider two sensor observation models with the observation
  process $\{\Obsi_k\}$ and $\{\Obsii_k\}$  generated by cdfs $ \cdf_1(\obs| \state)$ and $\cdf_2(\obs| \state)$,
  respectively.
  Suppose
  %\footnote{The partial order $\lSC$ denotes single crossing:  that
%$\cdf_1(\obs|\state) - \cdf_2(\bar{\obs}|\state)$ crosses zero from negative to positive at most once as $x$ increases for all $\obs,\bar{\obs}$.}
$ \cdf_1(\obs| \state) ,  \cdf_2(\obs| \state)$ satisfy a single crossing and 
  signed-ratio monotonicity condition (defined  in Sec.\ref{sec:assumptions}).
  Then convex stochastic dominance holds for the conditional mean: $\cmi (\Obsi_{1:k},\belief_0)   <_{cx}\cmii(\Obsii_{1:k},\belief_0)$, i.e., for any convex function $\fun: \reals \rightarrow \reals$ and prior $\belief_0$,
\begin{equation} \begin{split}
  \E_1\{ \fun\big( \cmi
  (\Obsi_{1:k},\belief_0) \big) \} \leq \E_2\{ \fun\big( \cmii
  (\Obsii_{1:k},\belief_0) \big) \}, \\  \text{ for all time $k$ }\label{eq:informal}
\end{split}
\end{equation}
Here $\E_\action$ denotes expectation wrt  the joint distribution of $\Obsa_{1:k}$.
  Therefore,\footnote{Note $\mse\{\cm(\Obs_{1:k},\belief_0)\} = \E\{\state^2\} -
    \E\{\cm^2(\Obs_{1:k},\belief_0)\}$. So clearly (\ref{eq:informal}) with
    $\fun(\cm) = \cm^2$ implies $\mse\{\cm_1(\Obs_{1:k},\belief_0)\} \geq \mse\{\cm_2(\Obs_{1:k},\belief_0)\}$.}
  $\mse\{\cm_1(\Obsi_{1:k},\belief_0)\} \geq \mse\{\cm_2(\Obsii_{1:k},\belief_0)\}$ for all
time $k$. 
  \label{thm:informal}
\end{theoremi}

 Theorem \ref{thm:informal}  says that localization using sensor 2 is always more accurate than using sensor 1 for any prior $\belief_0$ and this holds globally for all time $k$. To the best of our knowledge this result is  new. 
Convex stochastic dominance (\ref{eq:informal}) in a Bayesian framework has  been studied extensively in economics under the area of {\em integral precision dominance}, see \cite{GP10,AL18}. Theorem \ref{thm:informal} asserts convex dominance of the conditional mean $\cm_\action(\Obsa_{1:k},\belief_0)$ for all $k$, i.e., a global property. The proof involves combining two powerful results introduced recently in economics: integral precision dominance (which ensures that Theorem \ref{thm:informal} holds for $k=1$) and signed ratio monotonicity \cite{QS12}
 (which makes Theorem \ref{thm:informal} hold globally for all $k$). The usefulness of Theorem~\ref{thm:informal} stems from the fact that checking (\ref{eq:informal}) numerically is impossible since it  involves checking over a continuum of  priors and evaluating  intractable multidimensional integrals for  the  expected value.

The intuition behind (\ref{eq:informal}) is that of integral precision: if the observation is noisy, then the posterior is concentrated around the prior  while if the observation is more informative, then the posterior is more dispersed from the prior (large variance). This in turn implies that the noisy observation incurs a larger MSE.
In this paper, we show that Theorem  \ref{thm:informal} holds if
$\State\in \reals$ (scalar valued) or finite state.
Intuitively, if sensor 1 has a higher noise variance than sensor~2, then
(\ref{eq:informal}) holds - we will interpret this  in terms of
stochastic dispersion dominance in  Sec.\ref{sec:additive}.  But there are many other interesting cases where (\ref{eq:informal}) holds; 
the case with finite observation alphabets is particularly interesting,  since there is no noise variance interpretation in that case (the interpretation is in terms of Shannon capacity).
The single crossing assumption and signed monotonicity condition
in Theorem~\ref{thm:informal} are 
straightforward to check compared to the well known Blackwell dominance
\cite{WH80,RZ94}  (see Definition \ref{defn:blackwell}) which requires
factorization of probability measures; and they hold in several new examples
where Blackwell dominance does not.
For example, Blackwell dominance does not, in general, hold globally for all  $k$;
due to lack of commutativity of matrix multiplication.

{\em Applications in optimal filtering and controlled sensing}.
 Since  Theorem \ref{thm:informal} applies to any convex function, it has more applications than just characterizing the  mean square error of Bayesian localization.

As a first application, we will show that convex dominance applies to the one-step optimal (Bayesian) filtering update in a two-time scale model. That is,  consider a Markov process $\{\State_k\}$  which evolves over the slow time scale  $k$ with transition kernel $\State_{k+1} \sim \pdf(\state_{k+1}|\state_k)$,
and is  observed in noise via the observation process $\{\Obs_k\}$ at a fast time scale. So at each time  $k$, we obtain multiple fast time scale observations
 denoted as  the vector $\Obs_k = (\Obs_{k,1},\ldots,\Obs_{k,\lag})$ for some integer $\lag\geq 1$
where  each component
  $\Obs_{k,\lagc} \sim \cdf_\action(\cdot| \state_k)$ is conditionally independent of $\Obs_{k,m}$.
Then  one step of the optimal filter updates the posterior
  distribution  $\belief_k = \pdf(\state_k| \obs_{1:k})$ given $\belief_{k-1}$. 
Then the conditional mean is determined by $\obs_k$ and $\belief_{k-1}$ and denoted as $\cm(\obs_{k,1},\ldots,\obs_{k,\lag},\belief_{k-1})$.
\begin{theoremi}(Optimal Filtering).  \label{thm:informal2} Under the conditions of Theorem \ref{thm:informal},   convex dominance holds for the conditional mean for one  step of the optimal filter. That is, for any convex function
  $\belief:\reals\rightarrow \reals$ and  any prior $\belief_{k-1}$,
  %  the mean square error of the filtered posterior $\belief_k$ satisfies
     \begin{multline}  \E_1\{ \fun\big( \cmi
       (\Obs_{k,1},\ldots,\Obs_{k,\lag},\belief_{k-1}) \big)  \} \\
       \leq \E\{ \fun\big( \cmii
  (\Obs_{k,1},\ldots,\Obs_{k,\lag},\belief_{k-1}) \big) \} \quad \text{ for all $\lag$.} \label{eq:informalfilter} 
\end{multline}
Therefore $\mse\{\cmi(\Obsi_{1:k},\belief_{k-1})\}\geq \mse\{\cmii(\Obsii_{1:k},\belief_{k-1})\}$.
    \end{theoremi}
   Thus the optimal filter with sensor 2 is always more accurate than sensor 1. Similar to the discussion for Theorem \ref{thm:informal},  Theorem \ref{thm:informal2} is useful
    since in general, exact computation of the expectations is impossible
    (apart from the linear Gaussian case involving the Kalman filter).
The classical way of establishing (\ref{eq:informalfilter}) is Blackwell dominance. The main point is that  Theorem \ref{thm:informal2}
 covers several cases where Blackwell dominance does not hold (even for $\Delta = 1$).
Moreover, for $\Delta \geq 2$, in general Blackwell dominance will not hold. 
    
    As a second application, we will show that Theorem \ref{thm:informal2} is a crucial step in constructing myopic bounds for the  optimal policy
     of
controlled sensing partially observed Markov decision processes (POMDPs).
In controlled sensing POMDPs, the observation probabilities (which model an adaptive sensor) are controlled whereas  the transition probabilities  (which model the Markov  signal being observed by the sensor) are not controlled.
Controlled sensing arises in reconfigurable sensing resource allocation problems (how can a sensor reconfigure its behavior in real time), cognitive radio, adaptive radars and optimal search problems.
For such problems, the value function arising from stochastic dynamic programming  is convex but not known in closed form; nevertheless Theorem  \ref{thm:informal2} applies.
 By using Theorem \ref{thm:informal2}, the following useful structural result will
be established
\begin{theoremi} (Controlled Sensing POMDP)  For a 2-state Markov chain $\{\State_k\}$, under suitable conditions on the observation distributions, the optimal controlled sensing policy is lower bounded by a myopic policy. \label{thm:informal3}
\end{theoremi}

The motivation for Theorem \ref{thm:informal3} is two-fold.
First, since 
in general solving a POMDP for the optimal policy is computationally intractable, there is substantial motivation to derive structural results that bound the optimal policy; see 
\cite{Lov87,Rie91,Kri12,KP15,Kri16} for an extensive discussion of POMDP structural results and construction of myopic bounds. Second, the myopic bounds we propose are straightforward to compute and implement and can be used as an initialization for more sophisticated sub-optimal algorithms.
Existing works \cite{Rie91,Kri16}  in constructing myopic lower bounds to the optimal policy use Blackwell dominance of probability measures. Theorem \ref{thm:informal3}  includes several classes of POMDPs where  Blackwell dominance does not hold.

% Providing    useful sufficient conditions so that the optimal policy for a controlled sensing POMDPs is lower bounded by a myopic policy is surprisingly  nontrivial.

{\em Limitations}. Our results have two limitations.  First, for continuous-state problems, we require a scalar state ($\state\in\reals$). This is essential for convex dominance;  multivariate convex dominance is  an open area. Actually, for finite states (Hidden Markov Model  localization and filter) this is not a limitation  since a multivariate finite  state is straightforwardly mapped to a scalar finite state.
Second, while we show
 global convex dominance for localization (Theorem~\ref{thm:informal}), for optimal  filtering we can only show one-step convex dominance (Theorem \ref{thm:informal2}). Note however, Theorem \ref{thm:informal2} does hold for the two-time scale problem where the state remains fixed for multiple observations.   We emphasize that for the POMDP controlled sensing application, neither of these are limitations, since stochastic dynamic programming relies only on the one-step filtering update. Also, despite these limitations, the sufficient conditions given cover numerous  new examples where the only competing methodology (Blackwell dominance) does not hold.  Finally, using the ingenious proof of \cite{Den10}, it is  possible to give global convex dominance results for the optimal filter; but the corresponding sufficient conditions involve strong conditions and are complicated to check (albeit still finite dimensional); see Sec.\ref{sec:globalhmm} for a discussion.

{\em Related Works}.
As mentioned above,
  {\em Integral Precision dominance} which refers to convex dominance of conditional expectations,   has been studied in \cite{GP10}. The single crossing 
  condition proposed in \cite{Miz06} is a sufficient condition for integral precision dominance (for continuous-valued random variables observed in noise). Our main result, namely Theorem \ref{thm:main},  generalizes this to hold for an arbitrary sequence of observations - this requires generalizing
  the  single crossing condition of the observation probabilities in \cite{Miz06} to aggregating the single crossing condition \cite{QS12} and dealing with
  boundary conditions when the observation distribution has finite support.
  For a textbook treatment of convex dominance and stochastic orders in general, see \cite{MS02,SS07}.

  Regarding controlled sensing POMDPs,
\cite{WH80,Rie91,RZ94,Kri16} 
used  convexity of the value function together  with  Blackwell dominance
to construct a myopic lower bound. \cite{NAV13} considers controlled sensing with hypothesis testing.
  
% Thus far,  there has been no way of obtaining structural results for  controlled sensing POMDPs that exploit both
%monotonicity {\em and} convexity of the value function. The papers
% \cite{Lov87,Rie91} used only monotonicity of the value function (wrt monotone likelihood ratio stochastic order) but as mentioned above, it is impossible to generate non-trivial examples that satisfy the  resulting assumptions.

As mentioned earlier, Blackwell dominance
\cite{Bla53,WH80,RZ94}   requires
factorization of probability measures; and does not, in general, hold globally for all  $k$;
due to lack of commutativity of matrix multiplication. We refer the reader to \cite{Rag11,Rag16} for an excellent recent discussion on Blackwell dominance in an information theoretic setting. 
Finally, there are other approaches for quantifying the  MSE in estimation; \cite{SCC18} uses an interesting approach involving finite time anticipative rate distortion.

{\em Organization}.  Sec.\ref{sec:model} formulates the localization and filtering models, key assumptions, and main theorem (Theorem \ref{thm:main})  on convex dominance of the conditional mean. Sec.\ref{sec:examples} discusses important  examples where Theorem \ref{thm:main} applies including discrete memoryless channels, additive noise with log-concave density and  power law density. Sec.\ref{sec:pomdp} shows how local convex dominance of the optimal filter can be used to construct a myopic lower bound for the optimal policy of a controlled sensing POMDP.

\section{Convex Dominance for Bayesian Localization and Filtering}
\label{sec:model}
In this section we formulate the Bayesian estimation (localization and filtering problems),
and then present our main result on convex dominance of the conditional mean estimate, namely,
 Theorem \ref{thm:main}. The various assumptions required for Theorem \ref{thm:main} to hold are then discussed. Regarding notation, we use uppercase for random variables and lower case for realizations.
The superscript $^\p$ denotes transpose.

 \subsection{Bayesian Localization and Filtering Models}
 \label{sec:models}

For notational simplicity, we first formulate the filtering problem with
finite underlying state space $\statespace$. Then we formulate the continuous state filtering with state space on $\reals$. 
In either case, choosing the transition probability (density) as identity (Dirac mass) for the underlying Markov process  results in the Bayesian localization problem.

{\em Model 1. Finite State  Estimation}.
Consider a    discrete time Markov chain  $\{\State_k\}$ with finite state space $\statespace = \{1,2,\ldots, \statedim\}$, initial probability vector $\belief_0 = [\prob(\State_0  = 1),\ldots, \prob(\State_0  = \statedim)]^\p$ and  transition matrix $ \tp= \left[\tp_{ij}\right]_{\statedim\times\statedim}$,
$ \tp_{ij} = \prob(\State_{\time+1} = j | \State_\time = i )$. The Markov chain is observed in noise by sensor $\action$.
We 
consider
two sensors  $\action \in \{1,2\}$   which
generate the corresponding observation process $\{\Obsa_k\}$, 
 $k=1,2,\ldots$. Here 
 $\Obsa_k$ lies in observation space $\obspace_\action$ and has
 conditional distribution $ \cdf_\action(\cdot| \state_k)$, i.e., $\Obsa_k$ is conditionally independent
of $\Obsa_n$, $n< k$.
%At each time $k$, we obtain $\lag$ fast time scale observations; we represent these as the vector
 %$\obs_k = (\obs_k^1,\ldots,\obs_k^\lag)$ for some integer $\lag$.
%Each component
%  $\obs^\lagc_k \sim \cdf_\action(\cdot| \state_k)$ is conditionally independent %of $\obs_k^j$.
We consider three types of observation spaces  $ \obspace_\action$: either
$\obspace_\action$ is a finite set of
action dependent alphabets,
  $\obspace_\action =  \{1,2,\ldots,\obsdim_\action\}$,   $\action \in \actionspace$; or $\obspace_\action =  \reals$; or $\obspace_\action= [\obsl_\action,\obsu_\action]$, i.e.,
finite support
for  $\action \in \{1,2\}$.
Let
$\Belief = \left\{\belief: \belief(i) \in [0,1], \sum_{i=1}^\statedim \belief(i) = 1 \right\}$ denote the unit simplex of $\statedim$-dimensional probability vectors.

\begin{subdef}{defn}
\begin{defn}[Finite State Filtering and Localization] 
  % \begin{definition}[Finite State Filtering and Localization]
Assume $\tp,\cdf_\action(\cdot|\state),\belief_0$ are known.
Given an observation sequence $\obs_{1:k} = (\obs_1,\ldots,\obs_k)$ from sensor $\action$, 
the aim of filtering is to estimate the  Markov state $\State_k$, $k=1,2,\ldots$,  by computing
the posterior probability mass function  $\belief_{k} =  [\prob(\State_k=1| \obs_{1:k},\action),
 \ldots, \prob(\State_k = \statedim|\obs_{1:k},\action)]^\p \in \Belief$
recursively over time $k$.
Localization refers to the special case with transition matrix $\tp = I$ (identity matrix), and the aim is to estimate the random variable $\State_0$ by computing the posterior    
 $\belief_k  = [\prob(\State_0 = 1 | \obs_{1:k},\action),\ldots, \prob(\State_0 = \statedim|\obs_{1:k},\action) ]^\p\in \Belief$
recursively over time $k$. \label{def:filter}
%\end{definition}
\end{defn}

The solution to the filtering problem is as follows: Starting with initial distribution
 $\belief_0 = [\prob(\State_0=1),\ldots,\prob(\State_0 = \statedim)]^\p \in \Belief$,
the posterior using sensor $\action $ is computed recursively using  the classical hidden Markov model (HMM) state filter  as   
%Bayesian update
%$\belief_{k} = \filter(\belief_{k-1},\obs_k,\action)$ where
\begin{equation}
\begin{split} & \belief_{k} = \filter(\belief_{k-1},\obs_k,\action),  \text{ where } \filter\left(\belief,\obs,\action\right) = \cfrac{\oprob_{\obs} (\action)\, \tp^\p\,\belief}{\filterd\left(\belief,\obs,\action\right)} , \\
& \filterd\left(\belief,\obs,\action\right) = \one_{\statedim}'\oprob_{\obs}(\action) \tp^\p\,\belief, \\
&\oprob_{\obs}(\action) = \diag\{\oprob_{1,\obs}(\action),\cdots,\oprob_{\statedim,\obs}(\action)\}.
%\; \oprob_{i \obs}(\action) = \prod_{l=1}^\lag \oprob_{i \obs^l}(\action).
\end{split}
\label{eq:information_state}
\end{equation}
 Here $\one_{\statedim}$ represents a $\statedim$-dimensional vector of ones.
When the observation space $\obspace_\action$ of sensor $\action$ is a finite set,
$\oprob_{\state\obs}(\action) = \prob(\Obs_{\time+1} = \obs| \State_{\time+1} = \state, \action_{\time} = \action)$, $\obs \in \obspace_\action$ denotes the  observation   probabilities for sensor $\action$. When $\obspace_\action$ is continuum, we assume 
that the conditional distribution
$\cdf_\action(y| \state)$ is absolutely continuous wrt the Lebesgue measure and so the controlled conditional probability density function
$\oprob_{\state\obs}(\action) = \pdf(\Obs_{\time+1} = \obs| \State_{\time+1} = \state, \action_{\time} = \action)$ exists.
We assume  for each $y$,
$\oprob_{iy}(\action) \neq 0$ for at least one  state $i$; otherwise $\filterd(\belief,\obs,\action) = 0$ and $\filter(\belief,\obs,\action)$ are not well defined.

The notation in (\ref{eq:information_state}) specifies the filtering/localization update for a single observation $\obs_k$.
Given  a sequence of  observations $\obs_{1:k} = (\obs_1,\ldots,\obs_k)$ and prior
$\belief_0$, we denote the resulting computation of the  posterior $\belief_k$ 
as $\filter(\belief_0,\obs_{1:k},u)$ with normalization term
$\filterd(\belief_0,\obs_{1:k},u)$.
Let $\level = [\level(1),\ldots,\level(\statedim)]^\p$
denote the physical state levels associated with the states $1,\ldots, \statedim$,
respectively.
Then, for sensor $\action$, the  conditional mean estimate of the state is defined as the
$\Obsa_{1:k}$ measurable random variable
\beq \label{eq:cm}
 \cm_\action(\Obsa_{1:k},\belief_0) \ole \E_\action\{\level(\State_k)| \Obsa_{1:k},\belief_0\} =  \level^\p \,\filter(\belief_0,\Obsa_{1:k},\action) .\eeq
 Finally, for sensors $\action\in \{1,2\}$,  the mean square error (MSE) of the conditional mean given prior $\belief_0 $ is
\begin{multline} 
 \mse\{\cm_\action(\Obsa_{1:k},\belief_0)\} = \E\{ \big(\level(\State_k)- \cm_\action(\Obsa_{1:k},\belief_0)\big)^2 \}  \\ =
\E\{g^2(\State_k)\}-
\int_{\obspace_\action^k} \big( \cm_\action(\obs_{1:k},\belief_0)\big)^2\, \filterd(\belief_0,\obs_{1:k},\action) \,d\obs_{1:k}  \label{eq:mse}
\end{multline}
where $\int_{\obspace^k_u} $ denotes the $k$-dimensional integral over $ \obspace_u\times \cdots \times \obspace_u$.
 
 Given the complicated nature of (\ref{eq:cm}) and  (\ref{eq:mse}), evaluating  the MSE analytically for all
 priors $\belief_0$ is impossible, even when the observation space
 $\obspace_\action$ is finite. The MSE is  computed by Monte-Carlo simulation by averaging over a large number of sample paths $\obs_{1:k}$. Our main result below gives an analytical characterization for any convex function: given two sensors $\action \in \{1 ,2\}$,
with observation processes
% \footnote{Strictly speaking we should use the notation $\obs_k^\action$ since the sequence of random variables are generated from the conditional distribution $\cdf_u(y|x)$. But for notational convenience we use $\obs_k$.}
$\{\Obsi_k\}$,  $\{\Obsii_k\}$,   where observation $\Obsi \sim  \cdf_1(\cdot| \state)$ and  $ \Obsii \sim \cdf_2(\cdot| \state)$ respectively,
%which sensor yields a smaller MSE  for the conditional mean estimate?
 we give sufficient conditions so  that  $\mse\{\cmi(\Obsi_{1:k},\belief_0)\} \geq \mse\{\cmii(\Obsii_{1:k},\belief_0)\}$ for all priors $\belief_0$.

 {\em Model 2. Continuous State Estimation}: Here we assume a continuous state Markov process $\{\State_k\}$ with space   $\statespace = \reals$, initial distribution $\prob(\State_0 \in \borelset)$, and transition distribution $P(\State_{k+1} \in \borelset| \state_k)$ for any Borel set  $\borelset \subset\reals$. We assume absolute continuity wrt Lebesgue measure so that the initial density $\belief_0(x) = \pdf(\State_0 = \state)$ and transition density
 $\pdf(\state_{k+1}|\state_k)$ exists. The Markov process is observed by noise sensor $\action$. For each sensor $\action \in \{1,2\}$, we assume the
 observation space is  $\obspace_\action = \reals$. The observations
 are generated with conditional cdf $\cdf_\action(y|x)$ with support on $\reals$. We assume $\cdf_\action(y|x)$  is absolutely continuous wrt the Lebesgue measure and so the controlled conditional pdf 
$\oprob_{\state\obs}(\action) = \pdf(\Obs_{\time+1} = \obs|\State_{\time+1} = \state, \action_{\time} = \action)$ exists.

%\begin{definition}[Continuous State  Filtering and Localization]
\begin{defn}[Continuous State  Filtering and Localization]
Assume $\pdf(\state_{k+1}|\state_k), \cdf_\action(\obs|\state), \belief_0$ are known. Identical to Definition~\ref{def:filter} 
except that posterior $\belief_k = \pdf(\State_k= \state| \obs_{1:k},\action)$ is now a probability density function. In the localization problem, the transition density $\pdf(\state|\prevstate) = \delta(\state - \prevstate)$ is a Dirac mass. 
\end{defn} \label{def:filtering}
\end{subdef}
%\end{definition}

The solution of the filtering problem is as follows:
Starting with initial density $\belief_0(x)$,
 the posterior state density for sensor $\action$ is computed recursively using the optimal filter (Bayesian recursion)
 \begin{equation} \begin{split}
\belief_k(\state) &= 
\filter(\belief_{k-1},\obs_k,\action)(\state) ,\\
\text{ where } &  \filter(\belief,\obs,\action)(\state)= \frac{\oprob_{\state \obs}(\action)\, \int_\reals \pdf(\state| \prevstate) \belief(\prevstate) d\prevstate}{\filterd(\belief,\obs,\action)} , \\
 \filterd(\belief,\obs,\action) &= \int_\reals  \int_\reals \oprob_{\zeta \obs}(\action) \,\pdf(\zeta| \prevstate)  \belief(\prevstate) d\prevstate d\zeta.
\end{split}
\label{eq:contstate}
\end{equation}
The conditional mean estimate $\cm_\action(\Obsa_{1:k},\belief_0)$ of the state $\State_k$ and associated MSE for sensor $\action \in \{1,2\}$ are given by
\begin{equation} \begin{split}
\cm_\action(\Obsa_{1:k},\belief_0) = \E_\action\{\State_k | \Obsa_{1:k},\belief_0\} = \int_\reals \state \belief_k(\state) d\state, \\
  \mse\{\cm_\action(\Obsa_{1:k},\belief_0)\} = \E\{ \big(\State- \cm_\action(\Obsa_{1:k},\belief_0)\big)^2 \} 
\end{split}
\label{eq:contcm}
\end{equation}
Apart from the case where the densities $\pdf(\state_{k+1}|\state_k)$, $\cdf_\action(\obs|\state)$ and $ \belief_0$   are Gaussian\footnote{In the Gaussian case, posterior $\belief_k$ is Gaussian and its mean and variance are computed via the Kalman filter.},
 $\belief_k$ in  (\ref{eq:contstate}) does not have a finite dimensional statistic and can only
 be computed approximately (using, for example, sequential Markov-chain Monte-Carlo methods).
It is impossible to evaluate the MSE analytically over the continuum of  priors $\belief_0$; thus there is strong motivation to give sufficient conditions that yield convex dominance and therefore an ordering of the MSE between two sensor models $\action=1$ and $\action=2$.

{\em Remark. Two time scale filtering}: In Sec.\ref{sec:intro} we discussed a two time scale system where the state process $\{\State_k\}$ evolved on a slow time scale~$k$ and observations $\{\Obs_k\}$ are  recorded on a fast time scale. That is, at each time $k$ corresponding to state $\State_k$, we obtain $\lag$ fast time scale observations  represented by the vector
$\Obs_k = (\Obs_{k,1},\ldots,\Obs_{k,\lag})$ for some integer $\lag$ where
each component
  $\Obs_{k,\lagc} \sim \cdf_\action(\cdot| \state_k)$ is conditionally independent of $\Obs_k^j$. Then the filtering recursions (\ref{eq:information_state}) and (\ref{eq:contstate}) apply with $\oprob_{i\obs} (\action) =  \prod_{l=1}^\lag \oprob_{i \obs^l}(\action)$.

  \subsection{Assumptions and Main Result} \label{sec:assumptions}
We are now ready to state our main results.
  The key condition we will use is that of single crossing.
  
  \begin{defn}[Single Crossing \cite{Ami05}]  \label{def:sc}
    %For $z\in\reals$, define the signum function  $\sgn(z) \in \{-1,0,1\}$ for $z<0, z=0,z>0$, respectively.
  A function  $\fun: \statespace \rightarrow \reals$ is single crossing, denoted as $\fun(x) \in \lSC$ in $x\in \statespace$, if
\begin{equation}  \label{eq:single_crossing} \begin{split}
   % f(x) \lSC g(x) \iff
  \fun(x) \geq 0 \implies \fun(x') \geq 0 \text{ when } x' > x,
  % \text{ and } \fun(x') > 0 \implies \fun(x) > 0 \text{ when } x' > x
  \text{ and } \fun(x') \leq 0  \\
  \implies \fun(x) \leq 0 \text{ when } x' > x
  % \sgn ( f(x) - g(x) ) \text{ increases  in $x \in \statespace$. }
\end{split}
\end{equation}
\end{defn}
In words, $\fun(x)$ crosses zero at most once from negative to positive as $x$ increases. (Note that in our case $\statespace$ is a totally ordered set;
actually the single crossing definition applies more generally to partially ordered sets.)
  
  \subsubsection{Assumptions} The following are our main assumptions; recall
$\oprob_{\state\obs}(\action)$ is the conditional observation pdf, $\cdf_\action(\obs|\state)$ is the conditional observation cdf and $\bcdf_\action(\obs|\state)$ is the complementary  conditional cdf for sensor~$\action \in \{1,2\}$:
 \begin{enumerate}[label=(A{\arabic*})]
\item\label{TP2_obs} [TP2 observation probabilities] The observation probability kernel (matrix) $\oprob(\action)$ is totally positive of order 2 (TP2).
  \footnote{That is,
   $\oprob_{\state}(\action) \lr \oprob_{\bstate}(\action)$ where the monotone likelihood ratio (MLR) order $\lr$ is defined in Appendix~\ref{sec:proof}. Equivalently, for $\statespace$ finite,  the $i$-th row of $\oprob$  is MLR dominated by the $(i+1)$-th row, i.e., the rows of the matrix are totally monotone wrt the MLR order. When $\obspace_\action$ is finite, TP2  is equivalent to all second-order minors of matrix $\oprob(\action)$  being  nonnegative.}
 % \item\label{TP2_tp} [TP2 transition] $\tp$   is  totally positive of order 2 (TP2).

  \item \label{obssc} [Single Crossing Condition] For any $\bar\obs \in \obspace_1$, $\obs \in \obspace_2$,
    $ \cdf_1(\bar{\obs}| \state) - \cdf_2(\obs| \state) \in \lSC$ in $\state \in \statespace$.
    %i.e.,
    % $\cdf_1(\bar{\obs}|\state) - \cdf_2(\bar{\obs}|\state)$ crosses zero at most %once from  negative to positive as $\state$ increases.
    Equivalently, in terms of complementary cdfs,
 $ \bcdf_2(\obs| \state) - \bcdf_1(\bar{\obs}| \state) \in \lSC$.
%    
 % $\sum_{y\leq j} \oprob_{iy}(\action) - \sum_{y\leq \bj}  \oprob_{iy}(\action+1)$ changes sign at most once from negative to positive as $i$ increases for all $j \in %\obspace_\action,\bj \in \obspace_{\action+1}$. We denote this as
%  $\oprob(u+1) \gsc \oprob(u)$.
  % \item \label{obsdom} $\oprob_{iy}(u+1) \, \oprob_{jy}(u) \leq \oprob_{iy} (u) \, \oprob_{jy}(u+1)$, $j > i$.

\item \label{obsinit} [Boundary conditions]  %If $\obspace = \reals$, then $\oprob_{iy}(\action+1)/ \oprob_{iy}(\action)< \infty$
 % for $i=1,\ldots,\statedim$,
 % i.e., absolute continuity holds.
  If  $\obspace_\action= \{1,\ldots,\obsdim_\action\}$, $\action\in \{1,2\}$,
%  or $\obspace = [1,\obsdim]$ (finite support)
  then  for the boundary values $1$ and $\obsdim_\action$: %and  $i=1,\ldots,\statedim$:
\begin{equation*} \begin{split}
  \oprob_{\state1}(1) \, \oprob_{\bstate 1}(2)  &\leq
  \oprob_{\state 1}(2) \, \oprob_{\bstate 1}(1),
\\
  \oprob_{\state\obsdim_1}(1) \, \oprob_{\bstate \obsdim_{2}}(2)  &\geq
  \oprob_{\state\obsdim_{2}}(2) \, \oprob_{\bstate \obsdim_1}(1), \quad \bstate \geq \state.
\end{split}
\end{equation*}
If $\obspace_\action = [\obsl_\action,\obsu_\action]$  then the above equation %(\ref{eq:obsinit})
holds with $1$ and $\obsdim_\action$ replaced by $\obsl_\action$ and $\obsu_\action$.
\ref{obsinit} is not required if $\obspace_u = \reals$.
%\ref{obsinit}  holds trivially if density
% $\oprob_i(\action)$ is continuous (and therefore zero) at end points %$\obsl_\action$ and  $\obsu_\action$.

\item \label{sm}[Signed Ratio Monotonicity]
  If   $\bcdf_1(y|\state) < \bcdf_2(z|\state) $ and $\bcdf_1(\by|\state) > \bcdf_2(\bz|\state) $ then
 for all $y,\by \in \obspace_1$ and $\z,\bz \in \obspace_2$,
  $$ \frac{\log  \bcdf_1(y|\state)- \log \bcdf_2(z|\state)} { \log \bcdf_1(\by|\state) - \log \bcdf_2(\bz|\state)} \leq
  \frac{\log \bcdf_1(y|\bstate) - \log \bcdf_2(z|\bstate)} { \log \bcdf_1(\by|\bstate) - \log \bcdf_2(\bz|\bstate)} $$
for $ \bstate > \state$.\\
  If   $\bcdf_1(y|\state) > \bcdf_2(z|\state) $ and $\bcdf_1(\by|\state) < \bcdf_2(\bz|\state) $ then
  for all $y,\by \in \obspace_1$ and $\z,\bz \in \obspace_2$,
  $$ \frac{ \log \bcdf_1(\by|\state) - \log \bcdf_2(\bz|\state)} {\log  \bcdf_1(y|\state)- \log \bcdf_2(z|\state)}  \leq
  \frac { \log \bcdf_1(\by|\bstate) - \log \bcdf_2(\bz|\bstate)}{\log \bcdf_1(y|\bstate) - \log \bcdf_2(z|\bstate)},$$
for $ \bstate > \state$.
\end{enumerate}
The assumptions are discussed below in Sec.\ref{sec:discuss}. However, we note at this stage that 
  \ref{sm} is  equivalent to the following single crossing condition (proof in Theorem \ref{thm:agg} in the appendix):  for any $\obs_{1:k} \in \obspace_2^k$, $\bar{\obs}_{1:k} \in \obspace_1^k$ 
  \beq  %\big[\bcdf_2(\obs_1|\state) \cdots \bcdf_2(\obs_k| \state) \big]- 
  \prod_{l=1}^k \bcdf_2(\obs_l|\state) - \prod_{l=1}^k  \bcdf_1(\bar{\obs}_l|\state)
 % \bcdf_1(\bar{\obs}_1|\state) \cdots \bcdf_1(\bar{\obs}_k| \state) 
 \in \lSC, \; \state \in \statespace.
  \label{eq:sma}
  \eeq
  The main point is that (\ref{eq:sma}) globalizes \ref{obssc}, namely
$\bcdf_2(\obs| \state) - \bcdf_1(\bar{\obs}|\state) \in \lSC$,
  to a product from time 1 to arbitrary time $k$. \ref{sm} is a tractable condition for  (\ref{eq:sma}) in terms of the model parameters (observation probabilities); see discussion below.
  \\ \\
\subsubsection{Main result}

Our  main result involves  convex
dominance of the conditional mean. Let us define this formally.
\begin{defn}[Convex dominance of conditional mean]
  Consider two sensor  models $\action\in \{1,2\}$ with observation process $\{\Obsi_k\}$ and $\{\Obsii_k\} $  generated
  % \footnote{We emphasize again that  the sequence of random variables $\obs_k$ are generated from the conditional distribution $\cdf_u(y|x)$. For notational convenience we use $\obs_k$ instead of $\obs_k^\action$.}
  by cdfs $ \cdf_1(\obs| \state)$ and $\cdf_2(\obs| \state)$, respectively.
 Let $\int_{\obspace_\action^k}$ denote the $k$-dimensional integral over $\obspace_\action \times \cdots \times \obspace_\action$.
 Consider the  Bayesian localization/filtering problem of Definition \ref{def:filtering}.
 \begin{compactenum}
  \item   Global
    convex stochastic dominance of the conditional mean estimates (\ref{eq:cm}) or (\ref{eq:contcm})  denoted as   $\cmi (\Obsi_{1:k},\belief_0)   <_{cx}\cmii(\Obsii_{1:k},\belief_0)$ holds if for all time $k$,
 $\E_1\{ \fun\big( \cmi
  (\Obsi_{1:k},\belief_0) \big) \} 
 \leq \E_2\{ \fun\big( \cmii
 (\Obsii_{1:k},\belief_0)  \big) \} $  for any\footnote{Providing the integral exists.} convex function $\fun: \reals \rightarrow \reals$ and prior $\belief_0$.
Equivalently, for all time $k$,
\begin{multline}
\int_{\obspace_1^k} \fun\big( \cm_1(\obs_{1:k},\belief_0)\big)\, \filterd(\belief_0,\obs_{1:k},1) \,d\obs_{1:k}
\\ \leq    \int_{\obspace_2^k}\fun\big( \cm_2(\obs_{1:k},\belief_0)\big)\, \filterd(\belief_0,\obs_{1:k},2) \,d\obs_{1:k}
\label{eq:main} 
\end{multline}
\item      Local (one step) convex dominance of  the conditional mean estimates (\ref{eq:cm}) or (\ref{eq:contcm})  denoted as 
  $\cmi (\Obsi_{k},\belief_{k-1})   <_{cx}\cmii(\Obsii_{k},\belief_{k-1})$ holds at each  time $k$ 
if   $\E_1\{ \fun\big( \cmi
  (\Obsi_{k},\belief_{k-1}) \big) \} 
 \leq \E_2\{ \fun\big( \cmii
 (\Obsii_{k},\belief_{k-1})  \big) \} $ for
any convex function $\fun:\reals\rightarrow \reals$ and  prior $\belief_{k-1}$. Equivalently, at each time $k$,  %observation
 %   $\obs_k = (\obs_k^1,\ldots,\obs_k^\lag)$
\begin{multline}
\int_{\obspace_1} \fun(\cmi(\obs_k,\belief_{k-1})) \, \filterd(\belief_k,\obs_k, 1)\, dy_k \\  \leq
 \int_{\obspace_2} \fun(\cmii(\obs_k,\belief_{k-1})) \, \filterd(\belief_k,\obs_k, 2) \, dy_k  \label{eq:local}
\end{multline}
\end{compactenum}
\end{defn}

We are now ready to state our main results for Bayesian localization and filtering.

\begin{theorem}[Global Convex Dominance for Bayesian Localization]  
 Consider the  Bayesian localization problem of Definition~\ref{def:filtering}:
% defined in Remark 1, Sec.\ref{sec:models}:
 \begin{compactenum}
 \item  For 
    the finite state  model
    (\ref{eq:information_state}), under  \ref{TP2_obs},  \ref{obssc},  \ref{obsinit}, \ref{sm} (or (\ref{eq:sma})), global
    convex stochastic dominance of the conditional mean estimates (\ref{eq:cm}) holds for all time $k$, i.e.,   $\cmi (\Obsi_{1:k},\belief_0)   <_{cx}\cmii(\Obsii_{1:k},\belief_0)$.
    \item  For the continuous state  model   (\ref{eq:contstate}), under  \ref{TP2_obs},  \ref{obssc},   \ref{sm} (or (\ref{eq:sma})), global convex stochastic dominance  of the conditional mean estimates (\ref{eq:contcm}) holds
   for all time $k$.
 \end{compactenum}
  Therefore,  in  both cases, $$\mse\{\cm_1(\Obsi_{1:k},\belief_{0})\} \geq
   \mse\{\cm_2(\Obsii_{1:k},\belief_{0})\}$$  holds globally for all time $k$.
 \label{thm:main}
\end{theorem}

The proof of Theorem \ref{thm:main} is in Appendix~\ref{sec:proof}.

\begin{corollary}[Local Convex Dominance for Optimal Filtering]
Consider  the optimal filtering problem of Definition \ref{def:filtering}:
 \begin{compactenum}
 \item  For 
    the finite state  model
 under  \ref{TP2_obs},  \ref{obssc},  \ref{obsinit}, local convex dominance
 of the conditional mean estimates (\ref{eq:cm}) of the Hidden Markov Model (HMM) filter  (\ref{eq:information_state}) holds at each time $k$, i.e.,
 $\cmi (\Obsi_{k},\belief_{k-1})   <_{cx}\cmii(\Obsii_{k},\belief_{k-1})$.
\item For the continuous state model  under 
\ref{TP2_obs},  \ref{obssc}, local  convex dominance of  the conditional mean estimates (\ref{eq:contcm}) for  the optimal filter (\ref{eq:contstate}) holds at each time $k$.
\end{compactenum}
Therefore, for both cases,    $\mse\{\cm_1(\Obsi_{k},\belief_{k-1})\}  \geq
\mse\{\cm_2(\Obsii_{k},\belief_{k-1})$ holds  at each time $k$.
 \label{cor:main2}
\end{corollary}

\begin{corollary}[Two time-scale filtering]
  For the two-time scale filtering problem discussed in   Sec.\ref{sec:models},
  \begin{compactenum}
  \item For the HMM filter,
local convex dominance (\ref{eq:local}) holds
  under  \ref{TP2_obs},  \ref{obssc},  \ref{obsinit}, \ref{sm} 
  \item  For the continuous state filter,
local convex dominance (\ref{eq:local})  holds under \ref{TP2_obs},  \ref{obssc},  \ref{obsinit}, \ref{sm}.
  \end{compactenum}
  In either case, $\int_{\obspace_\action}$ in (\ref{eq:local}) denotes the $\lag$-dimensional integral. \label{cor:local}
\end{corollary}

\begin{proof} The one step filtering update (\ref{eq:information_state}) is identical to localization
  with  $\tp^\p \belief$ replaced by $\belief$. Since Theorem \ref{thm:main} holds for all $\belief \in \Belief$, Corollary \ref{cor:main2} and \ref{cor:local} follow.
\end{proof}

Let us reiterate the main point:
It is clear from  (\ref{eq:cm}), (\ref{eq:mse}) that evaluating  the MSE analytically for all
 priors $\belief_0$ is impossible, even when the observation space
 $\obspace_\action$ is finite. 
 Theorem \ref{thm:main} and its corollaries are useful since they give sufficient
 conditions that ensure one sensor observation model yields a MSE that dominates another sensor observation model; indeed they guarantee dominance
 for any convex function.
 Also for continuous state optimal filtering, in general there is no finite dimensional statistic for $\belief_k$ thereby making it impossible to compute exactly; yet Corollary \ref{cor:main2} and \ref{cor:local} give useful insight into how the observation probabilities affect the mean square error of the conditional mean.

\subsubsection{Why can't we establish global convex dominance of the optimal filter?}
The above results establish global convex dominance for Bayesian localization and local convex dominance for optimal filtering. 
 The key step in the proof of global convex dominance is (\ref{eq:lastline}) in the appendix: in simpler notation the task is to prove that 
$ (g-\lambda \ones)^\p \matd \,\belief\geq 0 $ for $\lambda \in \reals$ where $g$ is the vector of state levels of the Markov chain, $\matd$ is a square matrix, and $\belief$ is the prior. In the localization problem, $\matd$ is a diagonal matrix involving the observation distributions. Because of this diagonal structure, useful sufficient conditions can be given in terms of the model parameters $\oprob(1)$, $\oprob(2)$. In the filtering case $\matd$ is no longer a diagonal matrix - it is the  non-commutative product of transition matrices and observation matrices. Then there is no obvious way of giving useful sufficient conditions for $ (g-\lambda \ones) \matd \,\belief\geq 0 $ in terms
of the model parameters.

 In Sec.\ref{sec:globalhmm} we will give an alternative set of sufficient conditions
 for global convex dominance that apply to the optimal filter when the observation spaces $\obspace_\action$, $\action \in \{1,2\}$ are finite. However, checking these sufficient conditions for length $k$ observation sequences  requires a computational cost that is exponential in $k$ and so 
 intractable for large $k$. Nevertheless, the sufficient  conditions of  Sec.\ref{sec:globalhmm} guarantee global convex dominance for all (continuum  of) priors  $\belief$ and so are useful for small $k$.

 % Therefore,
 %   $\mse\{\cmi(\obs_{1:k},\belief_0)\} \geq \mse\{\cmii(\obs_{1:k},\belief_0)\}$.\\

    \subsection{Discussion of Assumptions \ref{TP2_obs}-\ref{sm}} \label{sec:discuss}
This subsection discusses the main assumptions of Theorem \ref{thm:main}.
Section \ref{sec:examples} below discusses several examples.

{\bf \ref{TP2_obs}}.
The TP2 condition
\ref{TP2_obs} is widely used to characterize the structural properties  of Bayesian estimation. \ref{TP2_obs} is  necessary and sufficient for the Bayesian filter update $\filter(\belief,\obs,\action)$ to be monotone likelihood ratio increasing
wrt $\obs$; see \cite{Kri16} for proof. This implies
$\cm_\action(\obs,\belief)$ is increasing in $\obs$.
This monotonicity wrt $\obs$ is a crucial step in proving Theorem \ref{thm:main}.
\cite{Kri16} gives several examples of continuous and discrete distributions that satisfy
\ref{TP2_obs} in the context of controlled sensing. We refer to the classic work \cite{KR80} for details and examples of  TP2 dominance, see also \cite{Kri16}.

{\bf \ref{obssc}}.
\ref{obssc} is the key condition required for {\em integral precision dominance}. First a few words about  integral precision dominance.
For random variable $x\in \reals$ with  prior $\belief$ and posterior $\filter(\belief,\obs,u)(x)$,   Definition 2(ii), pp.1011 in \cite{GP10} says that
integral precision dominance holds if the conditional expectations exhibit convex dominance:
\begin{multline*} m_1(\Obs) =  \int_\reals  x \filter(\belief,\Obs,1)(x) dx \\  \leq_{cx} m_2(\Obs) = \int_\reals x \filter(\belief,\Obs,2)(x) dx
\end{multline*}
Equivalently
$$\int_{\obspace} \fun\big( m_1(\obs)\big) \, \filterd(\belief,y,1)  dy \leq
\int_{\obspace}   \fun\big( m_2(\obs)\big) \, \filterd(\belief,y,2)  dy $$ for any convex function $\fun$,
providing the integrals exist. For $\state \in \reals$, \cite{Miz06} gives a single crossing condition similar to
\ref{obssc} for integral precision dominance; see also footnote 9, pp.1016 in \cite{GP10}. Our setting is different since  we consider a Markov process $\{\State_k\}$ observed in  noise
and we are considering convex dominance wrt the process  $\{\Obs_{k}\}$.
However, our main proof is similar in spirit to \cite{Miz06}, but in addition to \ref{obssc}, we also  need the boundary condition
\ref{obsinit} for finite support and finite set observations; also we need \ref{sm} for global convex dominance.
Finally, note  that \cite{Chi14} examines integral precision dominance 
as a special case of Lehmann precision (see Corollary 4.6 of \cite{Chi14}) after the seminal paper by \cite{Leh88}.
 
Returning to the single crossing condition \ref{obssc}, it
can also be viewed as signed-submodularity of the observation probability distributions. A function $\fun(\state,u)$ is  submodular if $\Delta(\state,\action) \ole \fun(\state,u) - \fun(\state,u+1)$ is  increasing in~$\state$.   In comparison, \ref{obssc} says\footnote{For $z\in\reals$, define the signum function  $\sgn(z) \in \{-1,0,1\}$ for $z<0, z=0,z>0$, respectively. Note that   $\sgn(\fun(x))$  increasing in $x$ (ignoring  excursions to zero) is equivalent to
  $\fun(x) \in \lSC$ in Definition \ref{def:sc}.}
$\sgn\big(\Delta(\state,\action)\big)$ is increasing in $\state$ where 
$\Delta(\state,\action) = \sum_{y\leq j} \oprob_{\state y}(\action) - \sum_{y\leq \bj}  \oprob_{\state y}(\action+1)$.  Requiring $\Delta(\state ,\action)$ to be increasing in $\state$ is impossibly restrictive, whereas 
requiring $\sgn\big(\Delta(\state ,\action)\big)$ to be  increasing in $\state$ leads to numerous  examples as discussed below.
We will use this signed-submodularity assumption in the FKG inequality 
(Theorem \ref{thm:convexdom}) to prove integral precision dominance.

{\bf \ref{obsinit}}. The boundary condition  \ref{obsinit} is not required if the observation space $\obspace_\action = \reals$ for $\action\in \{1,2\}$. \ref{obsinit} is only required when  $\obspace_\action$
 has  finite support  or $\obspace_\action$ is finite. \ref{obsinit} is not restrictive since it only  imposes   conditions on the observation probabilities at  the boundary values of $\obspace_\action$. %If $\obspace = [\obsl,\obsu]$
% for all $\action$ and $\oprob_{iy}(\action)$ is continuous (and zero) at
% $\obs = \obsl$ and $\obs = \obsu$, then \ref{obsinit} holds.
  \ref{obsinit} is a sufficient condition for the range of the posterior
for sensor  $1$ to be a subset of that
 for sensor  $2$, i.e.,
 $\{\level^\p \filter(\belief,\obs,1), \obs \in \obspace_1\}  \subseteq  \{ \level^\p \filter(\belief,\obs,2)$,  $\obs \in \obspace_{2}\}$. Several examples that satisfy \ref{obsinit} are given below.
Also to give further insight, the end of Appendix \ref{sec:finiteproof}  gives numerical examples where integral precision dominance does not hold
 when \ref{obsinit} is not satisfied.
 %A\ref{obsinit}   holds trivially when $\obspace = \reals$.

{\bf \ref{sm}}. {\em Signed ratio monotonicity}  \ref{sm} is a key condition
from the  paper \cite[Proposition 1]{QS12};  it
 is a necessary and sufficient for  any non-negative linear combination of single crossing functions to be single crossing.  Translated to our problem, 
\ref{sm} is required  for establishing Theorem \ref{thm:main} for $k>1$ (global convex dominance), i.e., when multiple observations $y_{1:k}$ are used to compute the posterior.
\ref{sm} is not required for the case $k=1$ (local convex dominance).
In simple terms \ref{sm} extends the single crossing condition  \ref{obssc} to the sum of single crossing functions. Note that  \ref{obssc} involves  each individual sensor $\action$, whereas \ref{sm} involves both sensors' observation probabilities.

%As mentioned earlier,  \ref{sm} is equivalent to (\ref{eq:sma}).
To motivate \ref{sm}, start with (\ref{eq:sma}). The ordinal property of single crossing \cite{Ami05} implies that (\ref{eq:sma}) is equivalent to
the difference in logs being single crossing, i.e.,
$\sum_{t=1}^k [\log \bcdf_2(\obs_t|\state)  - \log \bcdf_1(\obs_t|\state)] \in \lSC$.  Note \ref{obssc} implies that each
term $ [\log \bcdf_2(\obs_t|\state)  - \log \bcdf_1(\obs_t|\state)]  \in \lSC$; but this does not imply that the sum over $t$ is single crossing.
(In general the sum of single crossing functions is not single crossing.) The main point is that  signed ratio monotonicity condition \ref{sm} is necessary and sufficient for any non-negative linear combination of single crossing functions to be single crossing \cite[Proposition 1]{QS12}. This allows us to express (\ref{eq:sma}) as the tractable condition \ref{sm} which directly involves the observation density. Finally,  in the special case of additive log-concave noise densities, \ref{sm} automatically holds if \ref{obssc} holds; this is discussed below  in Sec.\ref{sec:additive}.

 Another intuitive way of viewing 
 (\ref{eq:sma}) is: a sufficient condition for local convex dominance
 is that   $ \bcdf_2(\obs| \state)- \bcdf_1(\bar{\obs}| \state)$ is increasing in $\state$ (this is stronger than \ref{obssc} which only needs
  $\sgn( \bcdf_2(\obs| \state)- \bcdf_1(\bar{\obs}| \state))$ to  increase in $x$);
 a sufficient condition for  global convex dominance requires that
 $ \bcdf_2(\obs| \state) /  \bcdf_1(\bar{\obs}| \state) $ is increasing in $\state$ (this is stronger than (\ref{eq:sma})).

 \section{Examples of Convex Dominance in Localization and Filtering} \label{sec:examples}
 To illustrate Theorem \ref{thm:main} and its corollaries, we  discus  3 important  examples of convex dominance in Bayesian estimation. Then we briefly discuss conditions for global convex dominance
 of the optimal filter.

 \subsection{Example 1. Blackwell Dominance,  Integral Precision Dominance and Channel Capacity}
 \label{sec:bdvsip}
 Here we discuss our first main example; namely how Theorem \ref{thm:main} and its corollaries apply to finite set observation models
 and HMMs.
As mentioned in Section \ref{sec:intro},  Blackwell dominance is a widely used   condition for convex dominance. 
Since Theorem~\ref{thm:main} uses integral  precision dominance to give a  new set of conditions for convex dominance compared to  Blackwell dominance, we compare them using several numerical examples below.

\begin{defn}[Blackwell dominance  $\oprob(2) >_{B} \oprob(1)$] \label{defn:blackwell}
Suppose $\oprob_{iy}(1) = \sum_{\bar{y}\in \obspace_2} \oprob_{i\bar{y}}(2) \, L_{\bar{y},y}$ for $y \in \obspace_1$ where $L$ is a stochastic kernel, i.e.,
$\sum_{y \in \obspace_1} L_{\bar{y},y} = 1$ and $L_{\bar{y},y} \geq 0$.
Then $\oprob(2) $ {\em Blackwell dominates} $\oprob(1)$; denoted as $\oprob(2) >_{B} \oprob(1)$. So when $\obspace_1,\obspace_2$ are finite,
$\oprob(2) >_{B} \oprob(1)$ if  $\oprob(1) = \oprob(2) \times L$ where $L$ is a stochastic 
(not necessarily square) matrix.
\end{defn}

Intuitively $\oprob(1)$ is noisier than $\oprob(2)$.
Using a straightforward  Jensen's inequality argument,  the following result holds:
\begin{theorem}[Blackwell dominance \cite{Rie91}]
  \label{thm:blackwell}
  %\begin{compactenum}
 %   \item  
      $\oprob(2) >_{B} \oprob(1)$ is  a  sufficient condition for the one step (local) stochastic dominance conclusion of Theorem \ref{thm:main} to hold.
   %    \end{compactenum}  
  \end{theorem}

  {\em Insight}.  Both integral precision dominance (Theorem \ref{thm:main}) and   Blackwell dominance (Theorem \ref{thm:blackwell})  exploit  convexity.  But there is an important  difference: Blackwell dominance implies that
  for any convex function $\fun: \reals^\statedim \rightarrow\reals$,
  $\sum_{\obspace_\action} \fun\big(\filter(\belief,\obs,\action)\big) \filterd(\belief,\obs,\action)$ is increasing in $u$ for all $\belief \in \Belief$. In comparison, integral precision dominance (Theorem \ref{thm:main}) implies convex dominance in one dimension, namely, for any scalar convex function
  $\fun: \reals \rightarrow \reals$,
   $\sum_{\obspace_\action} \fun\big(\level^\p\filter(\belief,\obs,\action)\big) \filterd(\belief,\obs,\action)$ is increasing in $u$
   for all $\belief \in   \Belief$.
   As will be shown below there any many important examples where integral
   precision dominance holds but Blackwell dominance does not hold.

   Note that Blackwell dominance (Theorem \ref{thm:blackwell}) does not hold globally for all $k$ unlike integral precision (Theorem \ref{thm:main}). This is because $\oprob(2) >_B \oprob(1)$ does not imply
   that the $k$-th powers satisfy $\oprob^k(2) >_B \oprob^k(1)$, apart from the pathological case
$\oprob(2) L = L \oprob(2)$ where matrix multiplication commutes (i.e., the pathological case when $L$ and $\oprob(2)$ are simultaneously diagonalizable).
Thus global convex dominance in Theorem \ref{thm:main} is a useful and substantial generalization.\footnote{Le Cam deficiency is a useful way of finding the closest Blackwell dominant  matrix to $\oprob(2)$ given $\oprob(1)$; it also yields the loss (deficiency) in choosing this closest  matrix, see \cite{Rag11} for a nice discussion.
  However, this loss is impossible to compute for an arbitrary convex function such as the value function of a controlled sensing POMDP which is apriori unknown and intractable to compute.}

\subsubsection*{Examples}

{\bf Example (i)}:  Here are  examples of  observation matrices that satisfy assumptions \ref{TP2_obs},   \ref{obssc},  \ref{obsinit}, \ref{sm} implying that  integral  precision dominance and   global convex dominance in Theorem \ref{thm:main} holds. But
 Blackwell dominance does not hold. 
%We now discuss examples where integral precision dominance holds but Blackwell dominance does not hold. 
%
 % {\em  Examples}.
 % (i) Here are  examples of  observation matrices that satisfy assumptions \ref{TP2_obs},   \ref{obssc},  \ref{obsinit} implying that   Theorem \ref{thm:main} holds:  %$ \statedim=3,\obsdim=3,\actiondim=2$,
  \begin{align*}
    \text{Ex1.} \; &
                  \oprob(1) = \begin{bmatrix}
                    0.8 & 0.2 & 0 \\
                    0.1 & 0.8 & 0.1 \\
                    0 & 0.2 & 0.8 
                  \end{bmatrix}, \;
                              \oprob(2) = \begin{bmatrix}
                                0.9 & 0.1 & 0 \\ 0.1 & 0.8 & 0.1 \\ 0 & 0.15 & 0.85
                              \end{bmatrix}
\\                             
                              \text{Ex2.} \;  &
  \oprob(1) = \begin{bmatrix}
     0.44847   & 0.30706 &  0.24447 \\
   0.33443  &  0.28762 &   0.37795 \\
   0.32463 &   0.28971 &   0.38565 
  \end{bmatrix}, \\ 
  & \qquad  \oprob(2) = \begin{bmatrix}
     0.170021  &  0.410485   & 0.419494 \\
   0.106500  &  0.433559  &  0.459941 \\
   0.020739 &   0.263223 &   0.716038
 \end{bmatrix}\\
    \text{Ex 3.} \; &
                           \oprob(1) = \begin{bmatrix}
                             0.8 & 0.2 \\ 0.2 & 0.8
                           \end{bmatrix},
                                                \; \oprob(2) = \begin{bmatrix}
                                                  0.7 & 0.3 & 0 \\ 0.1 & 0.2 & 0.7
                           \end{bmatrix}, \\ & \qquad  \obspace_1 = \{1,2\}, \; \obspace_2 = \{1,2,3\}.
  \end{align*}
 % Regarding the first example above,  the assumptions hold for arbitrary
 % $\statedim=\obsdim$,
  %$\oprob(1)$ and $\oprob(2)$ triadiagonal with 
  %For the second example above,  \ref{obsdom} also holds implying that statement 1 and statement 2 of Theorem~\ref{thm:main} hold.
  Note the third example has different observation spaces for the two actions.
Interestingly, in all three examples above, $\oprob(2)$ does not Blackwell dominate $\oprob(1)$; i.e.,
$\oprob(1) \neq \oprob(2) \times L$ for stochastic matrix $L$.
%This illustrates the usefulness of Theorem \ref{thm:main} compared to Theorem \ref{thm:blackwell}.
\\
%{\bf Example (ii)}. Suppose $\statedim=\obsdim$ are arbitrary positive integers; choosing either sensor 1 or sensor 2
%yields  a noisy  observation at most one unit different from  the Markov state, i.e.,
%$\oprob(1)$ and $\oprob(2)$  are tridiagonal matrices.
%Suppose sensor 1 is more accurate
%for states $2,,\ldots,\statedim-1$, while sensor 2 is more accurate for states 1 and $\statedim$.
%
%$\oprob(1)$ is specified as follows:
%$\oprob_{11}(1)=\oprob_{\statedim\statedim}(1) = p$, $\oprob_{12}(1) = \oprob_{\statedim,\statedim-1}(1) = 1-p$, $\oprob_{ii}(1) = p$, $\oprob_{i,i+1}(1)=\oprob_{i,i-1}(1) = (1-p)/2$.
%
%$\oprob(2)$ is specified as follows:
%$\oprob_{11}(2) = B_{\statedim\statedim}(2)= r$,
%$\oprob_{12}(2) = \oprob_{\statedim,\statedim-1}(2) = 1 - r$,
%  $\oprob_{ii}(2) = q$,
%$\oprob_{i,i+1}(2) = (1-p)/2$, $\oprob_{i,i-1}(2)= (1+p)/2 - q$ 
%
%Then for $r \geq p$, $q\geq p$,
%and  $p \in [1/(1+\sqrt{2}),1]$,
% \ref{TP2_obs},   \ref{obssc},  \ref{obsinit} hold 
%but Blackwell dominance  does not hold. %(verified numerically).
%\\
{\bf Example  (ii)}.  A consequence of 
 \cite{Jew07} is that for  symmetric $2 \times 2$ matrices $\oprob(1), \oprob(2)$, if $\oprob_{11}(1) \leq \oprob_{11}(2)$, then Blackwell dominance is equivalent to integral precision dominance \ref{obssc}. Then \ref{obsinit} automatically holds. This is easy to show, see  \cite{GP10}:  $\oprob(2) >_{B} \oprob(1)$ since $L = \oprob^{-1}(2) \oprob(1)$ is a valid stochastic matrix as can be verified by explicit symbolic computation.
\\
{\bf Example (iii). Channel Capacity}.
 Shannon  \cite{Sha58} establishes the following result in terms of channel capacity; see \cite{CT06} for a detailed exposition.
\begin{theorem}[\cite{Sha58}] If $\oprob(1) = M\,  \oprob(2)\,  L $ where $L$ and $ M$ are stochastic matrices,
then discrete memoryless channel $\oprob(1)$ 
  has a smaller Shannon capacity (conveys less information)
  than $\oprob(2)$. \label{thm:shannon}
\end{theorem}
Blackwell dominance $\oprob(1) = \oprob(2)\, L $  is a special case of Theorem \ref{thm:shannon} when $M = I$.
%Recall $\oprob(2)$ Blackwell dominates
%$\oprob(1)$ if 
%$\oprob(1) = \oprob(2)\, L $ where $L$ is a stochastic matrix.
However,
if  the multiplication order is reversed, i.e., suppose $\oprob(1) = M \,\oprob(2)$ where $M$ is a stochastic matrix, then even though $\oprob(1)$ is still more ``noisy'' (conveys less information according to Theorem \ref{thm:shannon}) than $\oprob(2)$, Blackwell dominance does not  hold.

Motivated by Theorem \ref{thm:shannon}, a natural question is: Does integral precision dominance and hence
Theorem~ \ref{thm:main} hold  for examples where $\oprob(1) = M \, \oprob(2)$ where $M$ is a stochastic matrix?
As an  example  consider
\begin{multline*}
 \hspace{-0.4cm}  \statedim=3,  \obsdim=3,\actiondim=2,\;
  \oprob(1) = \begin{bmatrix}
0.3229  &  0.4703   &  0.2068\\
    	0.2237 &   0.4902  &  0.2861\\
    	0.1587  &  0.4620 &   0.3793 \end{bmatrix},\\
      \oprob(2) = \begin{bmatrix}
0.4387   &  0.5190   &  0.0423 \\
    	0.2455  &   0.6625 &   0.0920 \\
    	0.0615   &  0.2829 &   0.6556
              \end{bmatrix}
            \end{multline*}
            Then there exists a stochastic matrix $M$ such that $\oprob(1) = M \, \oprob(2)$ but  Blackwell dominance does not hold since  $\oprob(1) \neq \oprob(2)\, L$ for stochastic matrix $L$.
  But \ref{TP2_obs}, single crossing condition \ref{obssc},  boundary  condition \ref{obsinit}, and signed ratio monotonicity \ref{sm} hold  for this example; therefore  Theorem \ref{thm:main} holds.

Further examples involving hierarchical sensing and word-of-mouth social learning are discussed in Section \ref{sec:pomdp}.
  
 {\bf  Summary}: This subsection discussed  several  examples where integral precision dominance and global convex dominance of the conditional mean holds but Blackwell dominance does not hold. The two specific cases we discussed  are:
  \begin{compactenum}
  \item $\oprob(1) =  M\,  \oprob(2)\,  L $ where $L$ and $ M$ are stochastic matrices,
  \item Blackwell dominance $\oprob(2) >_B \oprob(1)$ does not  imply global Blackwell dominance
$\oprob^k(2) >_B \oprob^k(1)$. In comparison, Theorem \ref{thm:main} gives conditions for
    which global convex dominance holds.
  \end{compactenum}
  
 \subsection{Example 2. Sensing in Additive Noise with Log-concave density}
 \label{sec:additive}
We now discuss  how Theorem \ref{thm:main} and its corollaries apply
 to  sensing in 
 additive noise, where the additive noise has a log-concave density. The main point is that for additive noise with log-concave density,
 higher differential entropy or variance of the additive noise is a necessary condition for the MSE of the Bayesian localization and filtered estimate to be higher. (Sec.\ref{sec:power} below shows that if the noise does not have a log-concave density, then
 higher differential entropy or variance is not a necessary condition).
 
 In the additive noise setting,  the sensor observation models are $\Obsa_k = \State_k + \Snoisea_k$,
 $\action \in \{1,2\}$. The additive  noise $\Snoisea_k$ is independent and identically distributed with a log-concave pdf   $p_\Snoise(\cdot|\action)$.
Recall \cite{BB05} that a log-concave density has the form $\pdf_\Snoise(\snoise) = \exp( \fun(\snoise))$ where $\fun$ is a concave function of $\snoise$.
There are numerous examples of log-concave densities:
normal exponential, uniform, Gamma (with shape parameter $\alpha>1$),
Laplace, logistic, Chi, Chi-squared, etc.

%Note that  for
 %additive noise,  the observation  likelihood is  $\oprob_{\state \obsa}(\action) = p_\Snoise(\obsa- \state|\action) \,I(\obsa\geq \state)$.
 We assume for $\action\in \{1,2\}$ that the  density $\pdf_\Snoise(\cdot|\action)$ has either support on $\reals$ (then \ref{obsinit} is not required) in which case $\oprob_{\state \obsa}(\action) = p_\Snoise(\obsa- \state|\action) $; 
 or $\pdf_\Snoise(\cdot|\action)$ has support on $\reals_+$ in which case
 $\oprob_{\state \obsa}(\action) = p_\Snoise(\obsa- \state|\action) \,I(\obsa\geq \state)$
 (then \ref{obsinit} holds straightforwardly; e.g. if $\state \in \reals_+$, then $\obsl_\action = 0$ in \ref{obsinit} and both sides of the first inequality in \ref{obsinit} are zero.)

%As shown in \cite[1.C.67, pp.67]{SS07}, assumption  \ref{TP2_obs} holds iff
%$\pdf_\Snoise(\cdot|\action)$ is a log-concave density.
 
 The following result
 characterizes the assumptions of Theorem \ref{thm:main}
 for additive noise models with a log-concave density.

 \begin{theorem}
   Consider  the additive noise sensing model  $\Obsa_k = \State_k + \Snoisea_k$, $\action\in \{1,2\}$ where the additive  noise $\Snoisea_k$ is independent and identically distributed
    with  pdf   $\pdf_\Snoise(\cdot|\action)$ and cdf $\cdf_\Snoise(\cdot|\action)$.
   Then:
   \begin{compactenum}
   \item \ref{TP2_obs} holds iff  $\pdf_\Snoise(\cdot|1)$
     and $\pdf_\Snoise(\cdot|2)$ are  log-concave densities.
   \item
     \ref{obssc} holds iff $\cdf_\Snoise(\cdot|1) >_D \cdf_\Snoise(\cdot|2) $ holds where $>_D$ denotes the dispersive stochastic order.\footnote{Cdf  $G$ dominates cdf $F$ wrt dispersive order, denoted $G>_D F$,  if $F^{-1}(\beta) - F^{-1}(\alpha) \leq G^{-1}(\beta) - G^{-1}(\alpha)$ for $0 < \alpha < \beta < 1$.}
   \item  (\ref{eq:sma}) or equivalently \ref{sm} holds  if  $\pdf_\Snoise(\cdot|1)$
     and $\pdf_\Snoise(\cdot|2)$ are  log-concave densities and
     $\cdf_\Snoise(\cdot|1) >_D \cdf_\Snoise(\cdot|2) $, i.e., \ref{obssc} holds.
      \item $\pdf_\Snoise(\cdot|2)$  having smaller  differential entropy than $\pdf_\Snoise(\cdot|1)$  is
     a necessary condition for \ref{obssc} to hold. Also  $\pdf_\Snoise(\cdot|2)$  having smaller  variance  than $\pdf_\Snoise(\cdot|1)$  is
     a necessary condition for \ref{obssc} to hold.
   \end{compactenum}
   Therefore for log-concave additive noise  $\pdf_\Snoise(\cdot|1)$ and $\pdf_\Snoise(\cdot|2)$, if  $\cdf_\Snoise(\cdot|1) >_D \cdf_\Snoise(\cdot|2) $, then Theorem \ref{thm:main}  and Corollaries \ref{cor:main2}, \ref{cor:local} hold.
     \label{thm:disp}
   \end{theorem}

   \begin{proof}
Statement 1 is proved in \cite[Theorem 1.C.52 (iii)]{SS07}.
 Statement 2   is proved
in \cite[Remark 3]{Miz06}. Statement 4 follows from \cite[Theorems 1.5.42 and 1.7.8]{MS02}.

Statement 3: Since the pdfs are log-concave, their
complementary cdfs $\bcdf_\Snoise(\snoise|1)$ and
$\bcdf_\Snoise(\snoise|2)$ are log-concave;
see \cite[Theorem 2(i)]{BB05}. Next from \cite[Theorem B 20. pp156]{SS07}, $\cdf_\Snoise(\cdot|1) >_D \cdf_\Snoise(\cdot|2) $ and
the complementary cdfs being log-concave implies that hazard rate dominance
$\cdf_\Snoise(\cdot|1) >_H \cdf_\Snoise(\cdot|2) $ holds, i.e.,
$\bcdf_\Snoise(\snoise|2) / \bcdf_\Snoise(\snoise|1)$ is decreasing in $\snoise$.
This implies $\bcdf_\Snoise(\bar{\obs}-\state|2) / \bcdf_\Snoise(\obs-\state|1)$ is increasing in $\state$ for all $\bar{\obs} \in \obspace_2$ and $\obs \in\obspace_1$. Therefore, $\log \bcdf_\Snoise(\bar{\obs}-\state|2) - \log
\bcdf_\Snoise({\obs}-\state|1) $ is increasing in $\state$ which in turn
implies that $\sum_{t=1}^k \log \bcdf_\Snoise(\bar{\obs}_t-\state|2) - \log
\bcdf_\Snoise({\obs}_t-\state|1) $ is increasing $\state$.
Therefore $\log \prod_{t=1}^k \bcdf_2(\bar{\obs}_t|\state) -
\log \prod_{t=1}^k \bcdf_1(\bar{\obs}_t|\state)$ is increasing in $\state$
which implies
 $\log \prod_{t=1}^k \bcdf_2(\bar{\obs}_t|\state) -
\log \prod_{t=1}^k \bcdf_1(\bar{\obs}_t|\state) \in \lSC$.
Finally,
$\fun_1(x) - \fun_2(x) \in \lSC$ implies that $\fun_1(f(x)) -  \fun_2(f(x)) \in \lSC$ for any monotone function\footnote{This is the well known ordinal property of single crossing \cite{Ami05}.} $f$.
Thus (\ref{eq:sma}) holds.
\end{proof}

Theorem \ref{thm:disp}
gives a complete characterization of global convex dominance in additive noise models.
It confirms the intuition that additive noise with higher differential entropy (or variance)  results in larger MSE for Bayesian localization and optimal filtering.  More precisely, higher differential entropy (or variance) is a necessary condition for \ref{obssc}; indeed \ref{obssc} (dispersion dominance)  is a stronger condition than dominance of differential entropy. 

 Examples  of log-concave densities that satisfy \ref{TP2_obs}, dispersive dominance \ref{obssc} and therefore \ref{sm} include:
 \begin{compactenum}
 \item
   Normal cdf:  $\cdf_\Snoise(\snoise|\action) = N(0,\sigma_\action^2)$ with $\sigma_1^2 \geq \sigma_2^2$, $\snoise \in \reals$.
 \item Exponential cdf: $\cdf_\Snoise(\snoise|\action) = 1 - \exp(-\lambda_\action \snoise) $, with rate parameter $\lambda_2 \geq \lambda_1$,
   $\snoise \in \reals_+$.
 \item Gamma distribution \cite{Sha82}: $\cdf_\Snoise(\snoise|\action) =\frac{1}{ \Gamma(\alpha_\action)} \snoise^{\alpha_\action-1} e^{-\snoise} $,  $\snoise \in \reals_+$ with shape parameter $\alpha_1 > \alpha_2\geq 1$.
   
\end{compactenum}
For these  examples Theorem \ref{thm:main} and Corollaries \ref{cor:main2}, \ref{cor:local} hold.
Also for these examples, \ref{obssc} is equivalent to $\pdf_\Snoise(\cdot|2)$  having smaller  differential entropy (or variance) than $\pdf_\Snoise(\cdot|1)$; that is Statement 4 of Theorem \ref{thm:disp}  is necessary and sufficient.

\subsection{Example 3. Additive Sensing. Power Law  vs Exponential Noise in Social Networks}
\label{sec:power}
Motivated by sampling social networks, we now discuss an example
where instead of the TP2 condition \ref{TP2_obs}, a reverse TP2 condition holds (due to  log convex density additive noise).
The main point below is that regardless of whether  the  power law noise has a smaller variance than exponential noise, the MSE is always larger due to convex dominance.
%A similar proof to Theorem
%\ref{thm:main} establishes global convex dominance.

Suppose we wish to compare the MSE of the conditional mean estimates when
the additive noise $\pdf_\Snoise(\snoise|1) $ is a log convex density that decays according to a power law
while  $\pdf_\Snoise(\snoise|2) $ is
an exponential  density (log-concave). That is:
  \begin{align*}
\pdf_\Snoise(\snoise|1)&= (\alpha-1)\, (1+\snoise)^{-\alpha}, \;  \\
\cdf_\Snoise(\snoise|1)&= 1 - (\snoise + 1)^{1 - \alpha},\;
\alpha > 1,  \quad \snoise \in \reals_+\\
\pdf_\Snoise(\snoise|2)&= \lambda \exp(- \lambda \snoise) ,\; \\
\cdf_\Snoise(\snoise|2)&= 1- \exp(-\lambda\snoise), \; \lambda >0,\;
\quad \snoise \in \reals_+
\end{align*}
For example, the empirical degree distribution (number of neighbors of per node normalized by the total number of nodes) of several social media networks such as Twitter \cite{KLP10}  have a power law with exponent $\alpha \in [2,3]$; while
social health networks in epidemiology have an exponential degree distribution. Based on observations obtained by  sampling individuals in the network and asking each such individual how many friends it has (degree), a natural question is: how accurate is the Bayesian conditional mean estimate for the average degree of the network?

\begin{theorem} \label{thm:power}
  Consider the additive noise model $\Obsa_k = \State_k + \Snoisea_k$,
  $\action  \in \{1,2\}$ where the additive  noise $\Snoisea_k$ is independent and identically distributed
  with  pdf   $\pdf_\Snoise(\cdot|\action)$.  Then  the conclusions of Theorem \ref{thm:main} hold for the following cases:
  \begin{compactenum}
\item Power law density
  $\pdf_\Snoise(\snoise|1) $ and exponential density $\pdf_\Snoise(\snoise|2) $
\item Power law densities  $\pdf_\Snoise(\snoise|1) $ and $\pdf_\Snoise(\snoise|2) $ with power law coefficients $\alpha_2 > \alpha_1$.
\end{compactenum}
\end{theorem}

Theorem \ref{thm:power}(1)  is interesting because it asserts convex dominance between two different types of  noise densities. It says that  the conditional mean  estimate in  additive  exponential  noise is always more accurate than that in power law noise. Interestingly, the variance
% entropy of a power law density can be smaller than an exponential density; for example when $\alpha = 2$, the entropy of power law pdf is $2\ln(2)$. The entropy of an exponential density is $1 - \ln(\lambda)$; so for $\lambda < 0.67$, a power law has a smaller entropy.
for a power law density can be smaller than that of an exponential density; for power law exponent $\alpha=3.1$, the variance is 17.35 which is smaller than the variance of an exponential for $\lambda<0.24$; yet the MSE of the conditional mean is larger in power law noise.  (Note for $\alpha \leq 3$, the power law variance
is not finite). Theorem \ref{thm:power}(2) is intuitive; a larger power law implies the density decays faster to zero; and therefore the MSE  is smaller.

\begin{proof} Statement (1):
  \ref{TP2_obs} holds for the observation likelihood $\oprob(2)$,  but \ref{TP2_obs} does not hold for  $\oprob(1)$. Instead $\oprob(1)$ satisfies a reverse TP2 ordering: $\oprob_{\state}(1) \gr \oprob_{\bstate}(1)$, $\state < \bstate$. Indeed,
  $\oprob_{\state\obs}(1)/\oprob_{\bstate,\obs} = (1+\obs-\bstate)^\alpha/(1+\obs-\state)^\alpha$ is increasing in $\obs$ for $\state  < \bstate$.
Then using a similar proof to Theorem \ref{thm:main}, global convex dominance holds if (recall $\lSC$ is defined in (\ref{eq:single_crossing})):
$$ \bcdf_2(\obs_1|\state) \cdots \bcdf_2(\obs_k| \state) -
\cdf_1(\bar{\obs}_1|\state) \cdots \cdf_1(\bar{\obs}_k| \state) \in \lSC, \; \state \in \statespace.$$
A similar proof to Theorem \ref{thm:disp} shows that the above condition holds because
$$\frac{\bcdf_2(\bar{\obs}|\state)}{ \cdf_1(\obs|\state)} = \frac{\bcdf_\Snoise(\bar{\obs}-\state|2)}{ \cdf_\Snoise(\obs-\state|1)}
= \frac{\exp(\lambda(x-\bar{\obs}))}{1 -  (\obs-\state-1)^{1-\alpha}}, \quad
\alpha > 1
$$
is increasing in $\state$ for all $\bar{\obs} > \state$ and
$\obs  > \state$.\\
Statement (2): Since $\oprob(1)$ an $\oprob(2)$ are reverse TP2, the global convex dominance condition becomes
$\cdf_2(\obs_1|\state) \cdots \cdf_2(\obs_k| \state) -
\cdf_1(\bar{\obs}_1|\state) \cdots \cdf_1(\bar{\obs}_k| \state) \in \lSC$. This holds because $(1 - (\bar{\obs} - \state + 1)^{1-\alpha_2})/
(1 - (\obs - \state + 1)^{1-\alpha_1})$  is increasing in $x$ for $\obs,\bar{\obs}>x$.
\end{proof}

 % Power law vs exponential distribution: $\cdf_\Snoise(\snoise|1) = \frac{\snoise}{1+\snoise}$, 
% $\cdf_\Snoise(\snoise|2) = 1 - e^{-\lambda \snoise}$ with $\lambda>1$,   $\snoise \in \reals_+$. Note the pdf $\pdf_\Snoise(\snoise|1) = (1+\snoise)^{-2}$, hence the term power law; its differential entropy is 2. Actually $\pdf_\Snoise(\snoise|1)$ is log-convex

  \subsection{Single crossing in conditional mean and Global Convex Dominance of  HMM filter}
  \label{sec:globalhmm}
  So far we used the single crossing  of the conditional distributions
\ref{obssc},  \ref{sm},
  to establish convex dominance.
We conclude this section by discussing an alternative condition based on  an ingenious result from \cite{Den10}; it uses single crossing of the conditional mean to establish global convex dominance of the conditional mean; but the conditions are computationally expensive to verify.

\begin{proposition}[{\cite[Proposition 2.1]{Den10}}]
  Suppose  $\cm_\action(\obs,\belief)$, $\action \in \{1,2\}$ is increasing in $\obs$ and $\cm_2(\obs,\belief) - \cm_1(\obs,\belief) \in\lSC$ in $\obs$.  Then convex dominance holds
  for the conditional means. \label{prop:denuit}
  (Recall $\lSC$ is defined in  (\ref{eq:single_crossing})).
\end{proposition}

% Proposition \ref{prop:denuit} applied to the HMM filter with finite observation space comprised of two symbols yields:
% \begin{proposition}
%   Suppose $\obspace_1 = \obspace_2 = \{1,2\}$. Then the conclusions of Corollary \ref{cor:main2} (local convex dominance) hold for the HMM filter under \ref{TP2_obs}, \ref{obsinit}.
% \end{proposition}
% \begin{proof}\ref{TP2_obs} implies $\cm_u(\obs,\belief)$ is increasing in $\obs$.  \ref{obsinit} implies $\cm_2(1,\belief) \leq
%   \cm_1(1,\belief)$  and  $\cm_2(2,\belief) \geq
%   \cm_1(2,\belief)$, implying that $\cm_2(\obs,\belief) -  \cm_1(\obs,\belief) \in \lSC$ in $\obs$.
%   Thus the conditions in Proposition \ref{prop:denuit} hold.
% \end{proof} 

 %Thus far, Theorem \ref{thm:main} asserts    only local (one-step) convex dominance for the optimal filter.
We now use   Proposition \ref{prop:denuit}  to establish global convex dominance for  the HMM filter (\ref{eq:information_state}); but the sufficient conditions given below are expensive to check  and only tractable for finite observation spaces $\obspace_1$ and $ \obspace_2$.

Note that $\obs_{1:k} \in \obspace^k$ with $\obsdim^k$ elements. Label the $\obsdim^k$ elements lexicographically and denote them as 
 $z\in \{1,2,\ldots, \obsdim^k\}$.
% For notational convenience we  consider 
% $\obspace_{\action}=\{1,2\}$ and convex dominance  $\cm_\action(\obs_{1:k};\belief)$ of the conditional mean of the HMM filter.
For $u \in \{1,2\}$ and each $i,j \in \statespace$, define the $\statedim\times \statedim$ matrices
 \begin{equation}
   \begin{split}
     \mato_\action(\obs_{1:k}) &= \prod_{t=1}^k  \tp\,\oprob_{\obs_t}(\action), \\
\matsii_\action(i,j,z,\bz) &= \mato_\action(\bz)\, (e_je_i^\p - e_i e_j^\p)\, \mato^\p_\action(z) \\ 
& \qquad + \mato_\action(z)\, (e_ie_j^\p - e_j e_i^\p)\, \mato^\p_\action(\bz)
,
\\
\matii(i,j,z) &= \mato_2(z)\,(e_je_i^\p - e_i e_j^\p) \, \mato^\p_1(z) \\
& \qquad + \mato_1(z)\, (e_ie_j^\p - e_j e_i^\p)\, \mato^\p_2(z)
\end{split}
\end{equation}

 We introduce the following assumptions for global convex dominance:
 \begin{enumerate}[label=(A{\arabic*}),resume]
   %\item \label{TP2_trans} The transition matrix $\tp$ is TP2.
\item\label{self_dom}  The matrices $\matsii_\action(i,j,z,\bz)$ are elementwise positive for all $\bz > z$, $j> i$, $i,j \in \statespace$.
\item \label{obssc2} The matrices $\matii(i,j,z)$ are elementwise negative
  for $z< z^*$ and positive for $z\geq  z^*$,   for all $j> i$, $i,j \in \statespace$, for some  $z^* \in \{1,\ldots,\obsdim^k\}$.

\end{enumerate}

\begin{theorem}
   Under \ref{self_dom} and \ref{obssc2}, global convex dominance
   $\cmi(\Obsi_{1:k},\belief_0) <_{cx}  \cmii(\Obsii_{1:k},\belief_0)$ holds
   for all priors $\belief_0$ for the HMM filter (\ref{eq:information_state}) with finite observation space $\obspace_\action$,
   $\action \in \{1,2\}$.
 \end{theorem}
 
 \begin{proof}
 \ref{self_dom} is
 sufficient for $\filter(\belief,z,\action) \lr \filter(\belief,\bz,\action)$
 for $z < \bz$; this can be verified from the definition of likelihood ratio dominance, namely
 $$ \frac{e_i^\p \mato_\action(z) \belief}{e_i^\p \mato_\action(\bz) \belief} \geq
 \frac{e_j^\p \mato_\action(z) \belief}{e_j^\p \mato_\action(\bz) \belief}, \quad j \geq i, \bz \geq z$$
 
   This in turn is sufficient for 
the first condition of Proposition \ref{prop:denuit}, namely, $\cm_\action(z,\belief)$ is increasing in $z$.

Similarly it can be shown that \ref{obssc2} is sufficient for 
  $\filter(\belief,z,1)\lr \filter(\belief,z,2)$ for $z\geq  z^*$ and 
    $\filter(\belief,z,1)\gr \filter(\belief,z,2)$ for $z<   z^*$. This implies
the second condition of Proposition \ref{prop:denuit} is satisfied, namely
 $\cm_2(z,\belief)\leq \cm_1(z,\belief)$, $z < z^*$ and
$\cm_2(z,\belief)\geq \cm_1(z,\belief)$, $z \geq  z^*$.
\end{proof} 

{\em Example}. It can be verified numerically that $\tp = \begin{bmatrix} 0.9 & 0.1 \\ 0.1 & 0.9\end{bmatrix}$,
$\oprob(1) = \begin{bmatrix} 0.7 & 0.3 \\ 0.3 & 0.7 \end{bmatrix}, \oprob(2) = \begin{bmatrix} 0.8 & 0.2 \\
0.2 & 0.8 \end{bmatrix}$ satisfies \ref{self_dom} and \ref{obssc2} for $k=1,2$.

{\bf Summary}: In contrast to previous subsections, this subsection used the single crossing property of conditional means to propose sufficient conditions \ref{self_dom} and \ref{obssc2} for global convex dominance of the HMM filter. Verifying \ref{self_dom} and \ref{obssc2} involve checking negative/positive elements for $O(\obsdim^{2k} \statedim^2)$ matrices is computationally intractable for large $k$. However, the conditions  guarantee global convex dominance for all (continuum) of priors  $\belief_0$ and are useful for small $k$.

  \section{Example. Controlled Sensing Partially Observed Markov Decision Process (POMDP)} \label{sec:pomdp}
  Thus far we have discussed convex dominance of the conditional mean  (in filtering and localization) between two fixed sensors. This section
 considers a POMDP controlled sensing  problem where we optimize the dynamic switching between multiple sensors.
  The main result of this section is
an important application of Corollary  \ref{cor:main2} (local convex dominance for the HMM filter): we
 construct a myopic lower bound to the optimal policy of a 2-state (but arbitrary observation space $\obspace_\action$) controlled sensing POMDP. Thus far, the only known way of constructing such lower bounds involved Blackwell dominance \cite{WH80,Rie91,Kri16}. The plethora of examples in
 Sec.\ \ref{sec:examples} where integral precision dominance holds (but Blackwell dominance does not), demonstrates the  usefulness of Theorem \ref{thm:main} in controlled sensing.
 
  In  controlled sensing,
 the aim is to dynamically decide which sensor (or sensing mode) $\action_k$ to choose at each time $k$ to optimize the objective defined in (\ref{eq:discountedcost}) below. In general, POMDPs are computationally intractable to solve (PSPACE complete).  Therefore,
from a practical point of view, constructing a myopic lower bound  is useful since  myopic policies are trivial to compute/implement in large scale POMDPs  and provide  a useful initialization for more sophisticated sub-optimal solutions.

\subsection{Controlled Sensing POMDP}
 We consider  an infinite horizon discounted reward controlled sensing POMDP.  It is customary to call the posterior $\belief_k$ as the ``belief''. A   discrete time two-state Markov chain  evolves with transition matrix $\tp$ on the  state space $\statespace = \{1,2\}$.
 So the belief space $\Belieft$ is a two-dimensional simplex, namely $\belief(1) + \belief(2) = 1$, $\belief(1),\belief(2) \geq 0$.
 Denote the
action space  as $\actionspace = \{1,2,\ldots,\actiondim\}$. For each action $\action
\in \actionspace$ denote the  observation space as $\obspace_\action$. We assume either  $\obspace_\action =  \{1,2,\ldots,\obsdim_\action\}$, i.e., finite set of action dependent alphabets for all $\action \in \actionspace$, or $\obspace_\action =  \reals$,  or $\obspace_\action= [\obsl_\action,\obsu_\action]$, i.e.,
finite support
for all $\action \in \actionspace$.
 For stationary policy  $\policy: \Belieft \rightarrow \actionspace$,
 initial belief  $\belief_0\in \Belieft$,  discount factor $\discount \in [0,1)$, define the  discounted cumulative reward:
\begin{align}\label{eq:discountedcost}
J_{\policy}(\belief_0) = \Ep\biggl\{\sum_{\time=0}^{\infty} \discount ^{\time}\, \reward_{\policy(\belief_\time)}^\p\, \belief_\time\biggr\}.
\end{align}
%that minimizes (\ref{eq:discountedcost}).
Here $\reward_\action = [\reward(1,\action),\reward(2,\action)]^\p$  is the reward vector for each sensing action $\action \in \actionspace$, and the belief state evolves according to hidden Markov model filter defined in (\ref{eq:information_state}) where
$\oprob_{\state\obs}(\action) = \prob(\obs_{\time+1} = \obs| \state_{\time+1} = \state, \action_{\time} = \action)$, $\obs \in \obspace_\action$ denotes the controlled observation   probabilities.

The aim is to compute the optimal  stationary policy $\optpolicy:\Belieft \rightarrow \actionspace$ such that
$J_{\optpolicy}(\belief_0) \geq J_{\policy}(\belief_0)$ for all $\belief_0 \in \Belieft$.
Obtaining the optimal stationary policy  $\optpolicy$ is equivalent to solving
 Bellman's  stochastic dynamic programming equation:
$ \optpolicy(\belief) =  \underset{\action \in \actionspace}\argmax~ Q(\belief,\action)$, $J_{\optpolicy}(\belief_0) = \valuef(\belief_0)$, where
\begin{multline}
\valuef(\belief)  = \underset{\action \in \actionspace}\max ~Q(\belief,\action), \\ 
  Q(\belief,\action) =  ~\reward_\action^\prime\belief + \discount\int_{\obspace_\action} \valuef\big(\filter(\belief,\obs,\action)\big)\, \filterd(\belief,\obs,\action)\,dy. \label{eq:bellman}
\end{multline}
% For continuum $\obspace_\action$,  $\sum_{\obs \in \obspace_\action}$ denotes integration wrt
%Lebesgue measure.
The value function  $\valuef(\belief)$ is the fixed point of the following value iteration algorithm: Initialize $V_0(\belief)=0$ for $\belief \in \Belieft$. Then
for $k=0,1,\ldots$ \begin{equation}
  \begin{split}
  \valuef_{k+1}(\belief)  &= \underset{\action \in \actionspace}\max ~Q_{k+1}(\belief,\action), \quad \mu_k^* = \argmax_{\action \in \actionspace} Q_k(\belief,\action), \\
  Q_{k+1}(\belief,\action) &=  ~\reward_\action^\prime\belief + \discount\int_{ \obspace_\action} \valuef_k\big(\filter(\belief,\obs,\action)\big)\,\filterd (\belief,\obs,\action)\, d\obs. \end{split} \label{eq:vi}
\end{equation}
The sequence $\{V_k(\belief), k=0,1,\ldots\}$ of value functions converges uniformly  to $V(\belief)$ on $\Belieft$ geometrically fast.
Since  $\Belieft$ is continuum, Bellman's equation \eqref{eq:bellman} and the value iteration algorithm (\ref{eq:vi}) do not directly translate into practical solution methodologies since they need to be evaluated at each $\belief \in \Belieft$.
Almost 50 years ago,  \cite{Son71} showed that when $\obspace_\action$ is finite, then for any $k$,
$\valuef_k(\belief) $ has a finite dimensional  piecewise  linear and convex characterization.  
% $\valuef_f(\belief) = \max_
Unfortunately, the number of piecewise linear segments can increase exponentially with the action space dimension
$\actiondim$ and double exponentially with time $k$.
Thus there is strong  motivation for structural results to  construct useful myopic lower  bounds   ${\policyl(\belief)}$ for the optimal policy $\optpolicy(\belief)$. %For belief states $\belief$ where $\policyl(\belief)= %\actiondim$, the optimal policy $\optpolicy(\belief)$ coincides with the %myopic policy $\policyl(\belief)$.

{\em Remark 1.}
For controlled sensing POMDPs, the transition matrix $\tp$, which characterizes the dynamics of the signal being sensed, does not depend on action $\action$. Only $\reward_\action$, which models the information acquisition reward of the sensor, and observation probabilities $\oprob(\action)$, which model the sensor's accuracy when it operates in mode $\action$, are action dependent.

{\em Remark 2.} A POMDP with finite horizon $\horizon$ has objective 
$ J_{\policy}(\belief_0) = \Ep\left\{\sum_{\time=0}^{\horizon-1}  \reward_{\policy_{\horizon-k}(\belief_\time)}^\p\, \belief_\time + \rstop^\p \belief_\horizon\right\}$ where $\policy=(\policy_1,\ldots,\policy_{\horizon})$ and $\rstop$ is the terminal reward vector. Then (\ref{eq:vi}) initialized as $\valuef_0(\belief) = \rstop^\p \belief$ for iterations  $k=0,\ldots,\horizon-1$ yields the optimal policy sequence $\policy^*=(\policy_1^*,\ldots,\policy_{\horizon}^*)$.

\subsection{Main Result -- Myopic Lower Bound}

  \begin{theorem}[Controlled sensing POMDP] \label{thm:pomdp}
         Assume      \ref{TP2_obs},   \ref{obssc},  \ref{obsinit}  hold.
Then      $
     Q(\belief,\action) - \reward_\action^\p \belief $ is increasing\footnote{By increasing, we mean non-decreasing.} in $ \action $.
     Therefore, the myopic policy $\policyl(\belief)= \argmax_\action \reward_\action^\p \belief$  forms a lower bound to the optimal policy in the sense that
     $\optpolicy(\belief) \geq \policyl(\belief)$ for all $\belief \in \Belieft$.
      Hence, for beliefs $\belief$ where $\policyl(\belief)= \actiondim$, the optimal policy $\optpolicy(\belief)$ coincides with the myopic policy $\policyl(\belief)$.
     An identical result holds in the finite horizon case for the policy sequence $\policy_k(\belief)$, $k=1,\ldots,\horizon$.
  % \item General  POMDP: Suppose both the transition probabilities  $\tp(\action)$ and observation probabilities $\oprob(\action)$ are action dependent.
 %    Then  under \ref{dec_cost}-\ref{obsinit}, the above result holds.
%   \end{compactenum}
   \end{theorem}
   \begin{proof}
The value function $\valuef(\belief)$ is convex in $\belief$ \cite{Kri16}.
Since $\statedim=2$, $\belief$ is completely specified by $\belief(2) = \level^\p \belief$ where $\level=[0, 1]^\p$. So $\valuef(\belief(2)) = \valuef(\level^\p \belief)$ is convex.
     Assuming
\ref{TP2_obs}, \ref{obssc}, \ref{obsinit}, it follows from
 Theorem \ref{thm:main} that 
 for all
$\belief \in  \Belieft$,
\begin{multline} \sum_{\obspace_{\action+1}} \valuef(\filter(\belief,\obs,\action+1) )\, \filterd(\belief,\obs,\action+1)  \\
\geq \sum_{\obspace_\action}  \valuef(\filter(\belief,\obs,\action) )\, \filterd(\belief,\obs,\action)
\label{eq:almostfinal} \end{multline}
Equivalently, see  (\ref{eq:vi}),
$ Q(\belief,u+1) - Q(\belief,u) \geq \reward_{u+1}^\p \belief - \reward_u^\p \belief $. Then Lemma 2 in \cite{Lov87} implies\footnote{
  Proof: If $u^* = \argmax_u \reward_u^\p \belief$, then (\ref{eq:almostfinal}) implies 
  $Q(\belief,u^*) \geq Q(\belief,u)$ for $u< u^*$. This implies $\mu^*(\belief) \in \{u^*,u^*+1,\ldots,\actiondim\}$. So $\policyl(\belief)
  = u^* \implies \mu^*(\belief) \in \{u^*,u^*+1,\ldots,\actiondim\}$.
  If $u^*$ is not unique, the proof needs more care, see Lemma 2,  \cite{Lov87}.
}
  $\mu^*(\belief) \geq \policyl(\belief)$ for all
$\belief \in  \Belieft$.
The same argument applies to $V_k(\belief)$ and  $\mu_k^*(\belief)$ for the finite horizon  case with terminal reward.
\end{proof}

From a practical point of view,   Theorem \ref{thm:pomdp} is useful since  the myopic policy $\policyl$ is trivial to compute and implement and gives a guaranteed lower bound to the optimal policy of the POMDP     which is intractable to compute.

The main point is that Theorem \ref{thm:pomdp}  provides an alternative to Blackwell dominance for POMDPs which has been widely studied since the 1980s  and also has the same conclusion:

\begin{theorem}[Blackwell dominance for Controlled Sensing. \cite{WH80,Rie91}]
  \label{thm:blackwellpomdp}
  %\begin{compactenum}
 %   \item  
      $\oprob(u+1) >_{B} \oprob(u)$, $u=1,\ldots,\actiondim-1$ is a  sufficient condition for the conclusion of Theorem \ref{thm:pomdp} to hold.
   %    \end{compactenum}  
  \end{theorem}
%  Theorem \ref{thm:main}  also holds for any finite horizon (with non-stationary policy).

Blackwell dominance holds for any number of states $\statedim$. In comparison Theorem \ref{thm:pomdp} applies only to POMDPs with 2 underlying states. However, there are numerous 2 state examples where Theorem \ref{thm:pomdp} applies and Blackwell dominance does not.
  
\subsection{Examples}

1. Theorem \ref{thm:pomdp} applies to all the 2-state examples in Sec.\ref{sec:examples} where \ref{TP2_obs},   \ref{obssc},  \ref{obsinit}  hold. As discussed in Sec.\ref{sec:bdvsip}, there are many examples  where Blackwell dominance does not hold, but integral precision dominance \ref{obssc} does hold.

2. In controlled radar sensing problems \cite{KD07},  observations are obtained at a faster time scale than the state evolution. That is, for state $\State_k$ (e.g., threat level at time $k$), an observation vector
$\Obs_k = (\Obs_{k,1},\ldots,\Obs_{k,\Delta})$ is obtained where $\Obs_{k,\lagc}$ and $\Obs_{k,m}$ are conditionally independent given $\State_k$. In such cases,
under \ref{sm}, convex dominance holds, and then Theorem \ref{thm:pomdp} holds. However, Blackwell dominance (Theorem \ref{thm:blackwellpomdp}) does not hold for this case.

 3.  Optimal filter vs predictor scheduling is an important application of controlled sensing.
 Filtering uses a sensor with observation  matrix $\oprob(2)$ to obtain measurements of the Markov chain and incurs a measurement cost but a performance reward. Prediction (no measurement) has non-informative  observation matrix $\oprob_{iy}(1) = 1/\obsdim$ and  incurs no measurement cost but yields a low performance reward. Clearly $\oprob(2) >_B \oprob(1)$.
If $\oprob(2)$ satisfies  \ref{TP2_obs}, then \ref{obssc} holds automatically because  $\sum_{y\leq j} \oprob_{iy}(1)$ is constant wrt $i$ ($\oprob(1)$ is non-informative), while \ref{TP2_obs} implies
$\sum_{y\geq l} \oprob_{iy}(2)$ is increasing wrt $i$.

 4.  Controlled Hierarchical Sensing:  In controlled sensing involving hierarchical sensors (such as hierarchical social networks),  level $l$ of the network
  receives signal $\state_k$ distorted  by the confusion matrix $M^l$  ($l$-th power of stochastic matrix $M$), where $l\in \{0,1,\ldots,\actiondim-1\}$.
That is,  each level of the network observes a noisy version of the previous level.
  Observing (polling) level $l$ of the network has  observation probabilities $\oprob$ conditional on the noisy message at level $l$. Therefore the conditional probabilities of the observation $\obs$ given the state $\state$ are
  $\oprob(U-l) = M^l \, \oprob(U)$ where $l$ is the degree of separation    from the underlying source (state).
  This is illustrated in Figure \ref{fig:hierarchical} for $\actiondim=3$.
  The controlled sensing POMDP is to choose which level to poll at each time in
  order to optimize an infinite horizon discounted reward.
  By Theorem \ref{thm:shannon},
$\oprob(u)$ is more noisy (has lower Shannon capacity) than $\oprob(u+1)$; yet 
Blackwell dominance does not hold due to the reverse multiplication order. But using integral  precision dominance,
Theorem \ref{thm:pomdp} holds
(under  assumptions).

\begin{figure} \centering
\tikzstyle{block} = [draw, rectangle, minimum height=2em, minimum width=4em]
\begin{tikzpicture}[thick,scale=0.75, every node/.style={transform shape},node distance = 1.5cm, auto]
  % Place nodes
%  \node [block] (block1) {Q};
  \node(block05){};
  \node[right of=block05](block15){};
   \node [block, right of=block15] (block2) {$M$};
    \node[ right of=block2] (block25) {};
    \node [block, right of=block25] (block3) {$M$};
    \node [right of=block3] (block35) {};
    \node[block,below of=block15,label=below:$\oprob(3)$](block4){$\oprob$};
    \node[block,below of=block25,label=below:$\oprob(2)$](block5){$\oprob$};
     \node[block,below of=block35,label=below:$\oprob(1)$](block6){$\oprob$};
  \draw[->] (block15) -- (block2);  \draw[->] (block2) -- (block3);
    \draw[->] (block15)--(block4);
    \draw[->] (block25)--(block5);
    \draw[->](block05)-- node[pos=0.5,above]{$\state_k\sim \tp$}(block15);
    \draw[->](block3)--(block35);
     \draw[<-](block6)--(block35);
   \end{tikzpicture}
   \caption{Controlled Hierarchical Sensing where Blackwell dominance does not necessarily hold.
     Level $l$ of the backbone network receives the Markovian signal $\state_k$ distorted by
   the confusion matrix $M^l$. Polling any specific level has observation probabilities $\oprob$; so the conditional probabilities of $\obs$ at level $l$ given $\state$ is specified by stochastic matrix $M^l \oprob$.}
   \label{fig:hierarchical}
\end{figure}
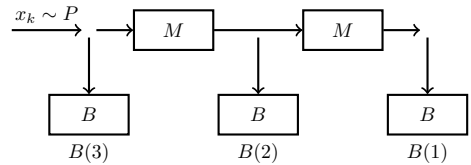

5. Word-of-Mouth Social Learning: Sensor 2 observes the Markov state in noise
with observation probabilities $\oprob(2)$. Sensor~1 receives the observations of sensor 1 in noise, but these probabilities also depend on the underlying state.
Denote these state dependent probabilities as  $M_i(l,m) \ole P(\Obsi_k = m| \Obsii_k= l,\State_k= i) $. Thus
 sensor 1 observation probabilities are
\beq \oprob_{im}(1) = P(\Obsi_k = m| \State_k = i) =
\sum_{l \in \obspace} \oprob_{il}(2)  \times M_i(l,m) \label{eq:wom} \eeq
% P(\Obsi_k = m| \Obsii_k= l,\State= i) $$
Such models arise in multi-agent social learning where agents use observations/decisions of previous agents and also their own private observations of the state to estimate the underlying  state \cite{Cha04,Kri12}.
Sensor 1 is influenced by the word-of-mouth message from sensor 2 but interprets (critiques) this message   based on its own observation of the state.
The controlled sensing problem involves dynamically choosing between sensor 1 (direct measurement from source) versus sensor  2 (word of mouth measurement) to optimize the cumulative reward (\ref{eq:discountedcost}).

 Even though from (\ref{eq:wom}), $\oprob(1)$ appears more noisy than $\oprob(2)$, 
 Blackwell dominance does not necessarily hold. Also  the Blackwell dominance proof of convex dominance breaks down due to the state dependent probabilities $M_i(l,m)$. However, integral precision dominance does hold in many cases. Here is one such example:
\begin{multline*}  \oprob(2) = \begin{bmatrix} 0.7  & 0.3 & 0 \\ 0.1 & 0.2 & 0.7 
\end{bmatrix},\;
M_1 = \begin{bmatrix} 0.9 & 0.1 \\ 0.5667 & 0.4333 \\0.2 & 0.8 
\end{bmatrix},\\
M_2 = \begin{bmatrix} 0.1 & 0.9 \\ 0.2 & 0.8 \\0.2143 & 0.7857
\end{bmatrix} ,  \;
 \oprob(1) = \begin{bmatrix} 0.8 & 0.2 \\ 0.2 &0.8 
\end{bmatrix}
\end{multline*}
It can be verified that  \ref{TP2_obs}, \ref{obssc}, \ref{obsinit} and \ref{sm} hold for this model, and therefore Theorem \ref{thm:main}
and Theorem \ref{thm:pomdp} hold.
\section{Discussion}
This paper developed sufficient conditions for local and global convex dominance of the conditional mean in Bayesian estimation (localization and filtering). We used two  techniques that have recently been developed in economics, namely, integral precision dominance (this yields local convex dominance)  and aggregating the single crossing property (this yields global convex dominance).  The convex dominance results  apply to several examples where Blackwell dominance does not hold. As an application, we showed how convex dominance can be used to construct  myopic lower bound to the optimal policy of a controlled sensing POMDP. The recent preprint  \cite{MPS19}  has interesting results on Blackwell dominance in large samples for two state random variables. In comparison the integral precision dominance used in the current paper yields global convex dominance  for an arbitrary number of states.

Our main result was to give  concise  sufficient conditions for global convex dominance in Bayesian localization (and for local convex dominance in  Bayesian filtering). In future work it is of  interest  to develop concise sufficient conditions for global convex dominance of Bayesian filtering; the conditions in Sec.\ref{sec:globalhmm} are difficult to verify. It is also worthwhile relating integral precision dominance (single crossing condition) to channel capacity. We know that Blackwell dominance $\oprob(2)>_B \oprob(1)$ implies that $\oprob(2)$ has higher capacity that $\oprob(1)$ (Theorem \ref{thm:shannon}). Since both Blackwell dominance and integral precision dominance imply convex stochastic dominance, giving sufficient conditions on integral precision dominance to relate to channel capacity provides useful links between the MSE of optimal filters, myopic policies of POMDPs and information theory.

Finally, this paper considered the effect of sensing (observation kernels) on  convex dominance and MSE when the transition kernels are identical. If the transition kernels are different for the two  observation processes, then the MSE of the conditional means are  meaningless since  the state processes are different. However, one can still establish local convex dominance of the optimal filter by introducing suitable conditions on the transition kernel.

\appendices
\section{Proof of Theorem \ref{thm:main}}  \label{sec:proof}

 \begin{defn}
Let $\belief_1, \belief_2 $ denote  two univariate pdfs (or pmfs).
Then $\belief_1$ dominates $\belief_2$ with respect to the monotone likelihood ratio (MLR) order, denoted as
$\belief_1 \gr \belief_2$,
 if 
 $ \belief_1(x) \belief_2(x') \leq \belief_2(x) \belief_1(x')$ for $x < x'$.
 \\
 $\belief_1$ dominates  $\belief_2$ with respect to first order dominance, denoted as
 $\belief_1\gs \belief_2$ if $\int_{-\infty}^x  \belief_1(\xi) d\xi\geq \int_{-\infty}^x \belief_2(\xi) d\xi$ for all $x$.
 A function $\fun:\belief\rightarrow \reals$ is said to be MLR (resp.\ first order) increasing if $\belief_1 \gr \belief_2$ (resp.\ $\belief_1 \gs \belief_2$) implies $\fun(\belief_1) \geq \fun(\belief_2)$. 
\end{defn}

For finite state space $\statespace$,
when  $\statedim =2$, $\gr$ is a complete order and coincides with
$\gs$.
For  $\statedim >2$, $\gr \implies \gs$ and both $\gr$, $\gs$ are partial
orders since it is not always
possible to order any two arbitrary beliefs  $\belief \in \Belief$.

Proceeding to the proof of Theorem \ref{thm:main},
for notational convenience we present the proof for  finite state space.
The proof for the continuous-state space case is virtually  identical and
outlined in Sec.\ref{sec:cont}. We assume that the state levels
 $\level$, associated with the state space $\statespace$, are ordered so that $\level_1< \level_2 < \cdots \level_\statedim$.
% Below $\belief(i)$ denotes the $i$-th element of belief $\belief \in \Belief$.

%  \begin{definition}
% Let $\belief_1, \belief_2 \in \Belief$ denote  two beliefs.
% $\belief_1$ dominates $\belief_2$ with respect to the monotone likelihood ratio (MLR) order, denoted as
% $\belief_1 \gr \belief_2$,
%  if 
%  $ \belief_1(i) \belief_2(j) \leq \belief_2(i) \belief_1(j)$ $i < j$,  $i,j\in \{1,\ldots,\statedim\}$.
%  \\
%  $\belief_1$ dominates  $\belief_2$ with respect to first order dominance, denoted as
%  $\belief_1\gs \belief_2$ if $\sum_{i\geq j} \belief_1(i) \geq \sum_{i\geq j} \belief_2(i)$ for $j \in \{1,\ldots,\statedim\}$.
%  A function $\fun:\Belief\rightarrow \reals$ is said to be MLR (resp.\ first order) increasing if $\belief_1 \gr \belief_2$ (resp.\ $\belief_1 \gs \belief_2$) implies $\fun(\belief_1) \geq \fun(\belief_2)$. 
% \end{definition}

% For state-space dimension $\statedim =2$, $\gr$ is a complete order and coincides with
% $\gs$.
% For  $\statedim >2$, $\gr \implies \gs$ and both $\gr$, $\gs$ are partial
% orders since it is not always
% possible to order any two arbitrary beliefs  $\belief \in \Belief$.

First note that the expectations of
$\cm_\action(\Obs_{1:k},\belief_0) $ are identical for  $u\in \{1,2\}$,
 because $\E\{  \cm_\action(\Obs_{1:k},\belief_0)    \}= \E\{ \E_{\action}\{\level^\p  x| \Obs_{1:k},\belief_0\} \} = \level^\p \E\{\state|\belief_0\}  = \level^\p \belief_0 $.
%First note that the expectations of $\level^\p \filter(\belief,\obs,\action) $ and
%$\level^\p \filter(\belief,\obs,\action+1) $ are identical, i.e.,
%$\sum_{\obspace_\action} \level^\p \filter(\belief,\obs,\action) \filterd(\belief,\obs,\action) = \sum_{\obspace_{\action+1}} \level^\p \filter(\belief,\obs,\action+1) \filterd(\belief,\obs,\action+1) = \level^\p \tp^\p \belief  $.
 Therefore Theorem 1.5.3
 in \cite{MS02} implies convex dominance is equivalent to increasing convex dominance.
Next, by
 Theorem 1.5.7 in \cite{MS02}, increasing convex dominance  holds iff for $\lambda \in \reals$,
\begin{multline}   \psi(\lambda) \ole \int_{ \obspace^k_{2}} [ \level^\p \filter(\belief,\obs_{1:k},2) - \lambda]^+ \filterd(\belief,\obs_{1:k},2) \, dy_{1:k}  \\
  - \int_{ \obspace^k_1}[\level^\p \filter(\belief,\obs_{1:k},1) - \lambda]^+ \filterd(\belief,\obs_{1:k},1)\, dy_{1:k} \geq 0. \label{eq:lambdafn}
\end{multline}
Here  we use the notation  $[x]^+ = \max(x,0)$. The remainder of the proof focuses
    on establishing (\ref{eq:lambdafn}).

    Defining $   \obspace^{k,\lambda}_u = \{\obs_{1:k}: \level^\p \filter(\belief,\obs_{1:k},u) > \lambda\}$,
% and $       \bar{\obspace}^{k,\lambda}_u = \obspace^k_\action - %\obspace^{k,\lambda}_u
% $,
\begin{align}
\psi(\lambda)    & = \int_{\obspace^{k,\lambda}_{2}}  [ \level^\p \filter(\belief,\obs_{1:k},{2})- \lambda]\, \filterd(\belief,\obs_{1:k},{2})\, d\obs_{1:k} \nonumber\\ & \quad  -
                   \int_{\obspace^{k,\lambda}_1} [ \level^\p \filter(\belief,\obs_{1:k},1) - \lambda] \, \filterd(\belief,\obs_{1:k},1)
                   \, d\obs_{1:k}
                   \label{eq:firstline}\\
  &= (\level-\lambda \ones)^\p \bigl[ \int_{{\obspace}^{k,\lambda}_2}
           \oprob_{\obs_k}(2)  \cdots  \oprob_{\obs_1}(2) \, d\obs_{1:k} \nonumber\\
         &\qquad \qquad \qquad - \int_{ {\obspace}^{k,\lambda}_{1}} \oprob_{\obs_k}({1}) 
\cdots  \oprob_{\obs_1}({1})  \, d\obs_{1:k}
             \bigr]  \,\belief  \nonumber\\
                 &=
                   (\level-\lambda \ones)^\p \bigl[ \bcdf_2(z_k) \bcdf_2(z_{k-1}) \cdots \bcdf_2(z_1) \nonumber\\ & \qquad \qquad \qquad -   \bcdf_1(\bz_k)  \bcdf_1(\bz_{k-1})  \cdots \bcdf_1(\bz_1)
                   \bigr]\,\belief  \label{eq:lastline} 
\end{align}
for some $z_1,\ldots,z_k \in \reals$ and $\bz_1,\ldots,\bz_k \in \reals$ which depend on $\lambda$.
Here each diagonal matrix $$\bcdf_u(z_i) = \diag[\bcdf_u( z_i|\state=1),\ldots, \bcdf_u( z_i|\state=\statedim)]$$ where
$\bcdf_u(z_i|\state) = 1 - \cdf_u(z_i|\state) $ is the complementary cdf.
Equation (\ref{eq:lastline}) follows since under \ref{TP2_obs}, $\filter(\belief,\obs_{1;k},u)$ is MLR increasing in each element $\obs_n$, $n=1,\ldots,k$. Therefore, the set
\beq {\obspace}^{k,\lambda}_\action
=  \{\obs_{1:k}: \level^\p \filter(\belief,\obs_{1:k},u) > \lambda\} = 
\{\obs_1 >\z_1, \ldots, \obs_k > \z_k \} \label{eq:olam}\eeq
for some $\lambda$ dependent real numbers $z_1,\ldots,z_k$ and hence  (\ref{eq:lastline})  involves complementary cdfs.

\subsection{Proof of Theorem \ref{thm:main} when $\statespace$ is finite and $\obspace = \reals$}

%Define
%$$ \E_{\filterd(\belief,u)} \{ \filter(\belief,\obs,u) \} = 

\begin{theorem}[Convex dominance for finite state localization]
  \label{thm:convexdom}
  Assume \ref{TP2_obs}, \ref{obssc}, \ref{sm} and $\obspace = \reals$.  Then 
  the following global convex dominance holds for all $k$:
  $\cm_1(\Obsi_{1:k}, \belief_0)  <_{cx} \cm_2(\Obsii_{1:k}, \belief_0)  $
  % $\level^\p  \filter(\belief,\cdot,1) <_{cx} \level^\p  \filter(\belief,\cdot,2)$.
  That is,
  for any convex function $\fun: \reals \rightarrow \reals$,
\begin{multline} 
  \int_{\obspace_1^k} \fun\big(\level^\p   \filter(\belief,\obs_{1:k},1)\big)  \, \filterd(\belief,\obs_{1:k},1)\, dy_{1:k} \\ \leq
  \int_{ \obspace_{2}^k} \fun\big(\level^\p   \filter(\belief,\obs_{1:k},2)\big)  \, \filterd(\belief,\obs_{1:k},2) \,dy_{1:k}.
\label{eq:convexdom}
\end{multline}
\end{theorem}

{\em Proof}.
   
%  {\bf Case 1. $\obspace = \reals$}: Here the observation conditional density $\oprob_{i}(\action)$ has
%  support on $\reals$ for each $i,\action$.
Since $\obspace = \reals$, clearly from (\ref{eq:firstline}), $\lim_{\lambda \rightarrow -\infty} \psi(\lambda) = \lim_{\lambda \rightarrow \infty} \psi(\lambda) = 0$.
  We establish (\ref{eq:lambdafn})  for $\lambda \in \reals$ by showing\footnote{Since $\psi(\lambda)$ is continuously differentiable (Lemma \ref{lem:nabla}) with $\psi(-\infty) = \psi(\infty) = 0$, clearly if $\psi(\lambda) \geq 0$ at its stationary points (minima), then $\psi(\lambda)\geq 0 $ for all $\lambda \in \reals$.} that $\psi(\lambda^*) \geq 0$
  at all stationary points $\lambda^*$ such that $d\psi(\lambda)/d\lambda = 0$.
Defining $\sgn(x) \in \{-1,0,1\}$ for $x<0, x=0,x>0$, respectively, (\ref{eq:lastline}) yields
   \begin{multline} \hspace{-0.5cm}
\psi(\lambda) = \sum_{i=1}^\statedim\underbrace{ (\level(i) - \lambda)}_{\alpha_i}\,
\underbrace{\sgn\biggl[ \prod_{t=1}^k \bcdf_2(z_t|x=i) - \prod_{t=1}^k \bcdf_1(\bz_t|x=i)  \biggr]}_{\beta_i}\, \\
\times \underbrace{\biggl|   \prod_{t=1}^k \bcdf_2(z_t|x=i) - \prod_{t=1}^k \bcdf_1(\bz_t|x=i) 
  \biggr|\,\,
\belief(i)}_{p_i}
% \end{split}
\label{eq:pre_fkg}
\end{multline}
Let us next evaluate the stationary points of  $\psi(\lambda)$ for $\lambda \in (0,1)$.

\begin{lemma} \label{lem:nabla}  $\psi(\lambda)$ defined
  in (\ref{eq:lambdafn}) is continuously differentiable  wrt $\lambda \in (0,1)$ with gradient
\begin{multline}
  \frac{d\psi(\lambda)}{d\lambda} = - \ones^{\p} \biggl[  \bcdf_2(z_k) \bcdf_2(z_{k-1}) \cdots \bcdf_2(z_1)  \\ -   \bcdf_1(\bz_k)  \bcdf_1(\bz_{k-1})  \cdots \bcdf_1(\bz_1) \biggr]  \,\belief  \label{eq:gradientexp}
\end{multline}
\end{lemma}
(Proof at the end of this subsection). \qed

Thus  the stationary points of $\psi(\lambda)$ satisfy
(using the notation  $\beta_i$, $p_i$ defined in (\ref{eq:pre_fkg}))
\begin{multline} \frac{d\psi(\lambda)}{d\lambda} = \ones^\p  \biggl[  \bcdf_2(z_k) \bcdf_2(z_{k-1}) \cdots \bcdf_2(z_1) \\ -   \bcdf_1(\bz_k)  \bcdf_1(\bz_{k-1})  \cdots \bcdf_1(\bz_1)\biggr] \, \belief = 
 \sum_i \beta_i p_i = 0. \label{eq:stationarypt}
\end{multline}
So to prove Theorem \ref{thm:convexdom},  it only remains to show that  $\psi(\lambda)$ is non-negative  at these stationary points.  To establish
this  we use the Fortuin-Kasteleyn-Ginibre (FKG) inequality \cite{FKG71} on (\ref{eq:pre_fkg}). In our framework the FKG inequality\footnote{Proof: Since $\alpha$ and $\beta$ are increasing vectors,  therefore $ (\alpha_i - \alpha_j)
  (\beta_i - \beta_j) \geq  0$ for all $i,j$. This  implies the expectation
$ \sum_i \sum_j (\alpha_i - \alpha_j) (\beta_i - \beta_j) p_i p_j \geq 0$ which immediately   yields the inequality (\ref{eq:fkg}).}
  reads: If $\alpha$, $\beta$ are generic increasing vectors
  and $p$ a generic probability mass function, then
  \beq \label{eq:fkg}
 \sum_i \alpha_i \beta_i p_i  \geq \sum_i \alpha_i p_i\,  \sum_j \beta_j p_j . \eeq
 Clearly in (\ref{eq:pre_fkg}):
 \begin{compactenum}
 \item
   $\alpha_i = \level(i) - \lambda$ is increasing since the elements of $\level$ are increasing by assumption;
   \item $\beta_i$ is increasing by Theorem~\ref{thm:agg} below;
   \item $p_i$ is non-negative and thus proportional to a probability mass function.
   \end{compactenum}
   Also from (\ref{eq:stationarypt}), $\sum_i \beta_i p_i = 0$.  So, applying  FKG inequality to
(\ref{eq:pre_fkg}) yields $\psi(\lambda) =  \sum_i \alpha_i \beta_i p_i \geq 0$.  Thus we have established~(\ref{eq:lambdafn}) for $\obspace = \reals$. \qed

 {\bf Proof of Lemma \ref{lem:nabla}}
  Here we prove Lemma \ref{lem:nabla} that was used to evaluate the gradient of $\psi(\lambda)$ in the proof above.
 For $s\in \reals$, similar to (\ref{eq:olam}) define $\obspace^{k,s}_\action = \{\obs_{1:k}: \level^\p \filter(\belief,\obs_{1:k},\action) > s\}$.   Start with (\ref{eq:lambdafn}), and use the so called ``integrated survival function'' on page 19, \cite{MS02}, namely, integration by parts yields
  $\int_{\obspace_\action^k} |\level^\p \filter(\belief,\obs_{1:k},u) - \lambda|^+ \filterd(\belief,\obs_{1:k},u) \, d\obs_{1:k}= \int_{\lambda}^\infty \int_{ {\obspace_\action^{k,s}}}
\filterd(\belief,\obs_{1:k},u) \, d\obs_{1:k}\, ds$.
Therefore 
$
\psi(\lambda) = \int_\lambda^\infty \ones^\p \bigl[ \int_{{\obspace}^{k,t}_2} \oprob_{\obs_k}(2) \cdots \oprob_{\obs_1} (2) \, d\obs_{1:k} -
\int_{{\obspace}^{k,t}_{1}} \oprob_{\obs_k}(1)\cdots \oprob_{\obs_1}(1)
\, d\obs_{1:k}\bigr] \, \belief\, dt $.
\begin{comment}
  For $t\in \reals$, define  $\obspace^t_u = \{\obs: \level^\p \filter(\belief,\obs,u) > t\}$.   Start with (\ref{eq:lambdafn}), and noting that
  $\sum_\obs |\level^\p \filter(\belief,\obs,u) - \lambda|^+ \filterd(\belief,\obs,u)
  = \int_\lambda^\infty |t- \lambda|^+ \sum_y I(\level^\p \filter(\belief,\obs,u) \geq t) dt$,
  we have
$$
\psi(\lambda) =
\int_{\lambda}^\infty |t-\lambda|^+\, \bigl[ \sum_{\obs \in \obspace^t_{\action+1}} \filterd(\belief,\obs,\action+1) - \sum_{\obs \in \obspace^t_\action} \filterd(\belief,\obs,\action)
  \bigr] dt
=\int_\lambda^\infty \ones^\p \bigl[ \sum_{\obs  \in \bar{\obspace}^t_\action} \oprob_y(\action) -
  \sum_{\obs  \in \bar{\obspace}^t_{\action+1}} \oprob_y(\action+1) \bigr] \tp^\p \belief\, dt $$
where the second equality follows since 
$\int_\lambda^\infty f(t) g(t) dt = f(\infty) g(\infty) - f(\lambda) g(\lambda) -
\int_\lambda^\infty g(x) d f(x) $ for generic $f,g$. 
%This expression
%for $\psi(\lambda)$ is  the integrated survival function,  \cite{MS02}.
\end{comment}
Then evaluating $d\psi(\lambda)/d \lambda$ yields (\ref{eq:gradientexp}).  Finally,  (\ref{eq:gradientexp}) implies $\psi(\lambda) $ is continuously differentiable
because $\sum_{\obspace_u^\lambda}\oprob_y(u)$ is continuous wrt $\lambda$ (since  $\oprob_y(u)$ is absolutely continuous wrt Lebesgue measure by assumption.)
\hfill \qed

\begin{theorem} \label{thm:agg}
  Under \ref{obssc} and \ref{sm},
$$ \beta_i 
  = \sgn[\prod_{t=1}^k \bcdf_2(z_t|x=i) - \prod_{t=1}^k \bcdf_1(\bz_t|x=i) ]$$
  in (\ref{eq:pre_fkg}) is increasing in $i$. (This property was used to  prove
  Theorem \ref{thm:convexdom}).
\end{theorem}
\begin{proof}
  Showing that $\beta_i$ is increasing in $i$ is equivalent to showing that
  $\prod_{t=1}^k \bcdf_2(z_t|x=i) - \prod_{t=1}^k \bcdf_1(\bz_t|x=i) $ is a single crossing function in $i$.
By the  ordinal property of single crossing \cite{Ami05},
this in turn is equivalent to showing that
  $\log \prod_{t=1}^k \bcdf_2(z_t|x=i) - \log \prod_{t=1}^k \bcdf_1(\bz_t|x=i) $
  is a single crossing function in $i$ or equivalently,
$ \sum_{t=1}^k [ \log  \bcdf_2(z_t|x=i) - \log \bcdf_1(\bz_t|x=i)]$
is single crossing.

\ref{obssc} implies that each term $\log  \bcdf_2(z_t|x=i) - \log \bcdf_1(\bz_t|x=i)$ is single crossing.
So proving $\beta_i \uparrow i$ boils down to showing that the sum of single crossing functions is single crossing. The paper  \cite{QS12} shows that the signed monotonicity ratio \ref{sm} is a necessary and sufficient condition for this to hold.
\end{proof}

\subsection{Proof of Theorem \ref{thm:main} when  $\statespace$ is finite and
  $\obspace_\action$ has finite support or $\obspace_\action$ is finite}
\label{sec:finiteproof}

The following result is required for establishing our main result when
$\obspace_\action$ is either finite or has finite support; see Case 1 and Case 2 of proof of Theorem \ref{thm:convexdom} below. It is here that \ref{obsinit} is the crucial assumption.

\begin{theorem}[Finite support observation distributions]
  \label{thm:boundary} Suppose $\obspace_\action = [\obsl_\action,\obsu_\action]$, $\action\in \{1,2\}$.
  Assume  \ref{TP2_obs}, \ref{obsinit}. Then $$\{\level^\p \filter(\belief,\obs,1), \obs \in \obspace_1\} \subseteq \{\level^\p \filter(\belief,\obs,2), \obs \in \obspace_{2}\}$$ Thus,
  defining   $\obspace^\lambda_u = \{\obs: \level^\p \filter(\belief,\obs,u) > \lambda\}$ and
$\bar{\obspace}^\lambda_u = \{\obs: \level^\p \filter(\belief,\obs,u) \leq \lambda\}$, it follows that
  $\obspace^\lambda_{2} = \emptyset, \obspace^\lambda_1 \neq \emptyset$  and
$\bar{\obspace}^\lambda_{2} = \emptyset, \bar{\obspace}^\lambda_1 \neq \emptyset$
are impossible,    
\end{theorem}
\begin{proof}
Since $\filter(\belief,\obs,u)$ is MLR increasing wrt $\obs$  under \ref{TP2_obs}, it suffices to show that
\beq \label{eq:boundaryineq}
\level^\p \filter(\belief,\obsl_{2},2) \leq \level^\p \filter(\belief,\obsl_1,1),
\quad \text{ and } \quad 
\level^\p \filter(\belief,\obsu_{2},2) \geq \level^\p \filter(\belief,\obsu_1,1) \eeq
The first inequality in (\ref{eq:boundaryineq}) is equivalent to
$ \sum_{i=1}^\statedim \sum_{j=1}^\statedim \level_i \big(\oprob_{i1}(2) \oprob_{j1}(1) -
\oprob_{i1}(1) \oprob_{j1}(2)\big) \belief(i) \belief(j) \leq 0$. So
 \ref{obsinit} is a sufficient condition for the inequality  to hold.
A similar proof holds for the second inequality  in (\ref{eq:boundaryineq}). \end{proof}

{\bf Case 1. $\obspace_\action=[\obsl_\action,\obsu_\action]$}:
Next we prove~(\ref{eq:lambdafn}) for  the finite support case where 
$\obspace_\action$ is the interval $[\obsl_\action,\obsu_\action]$.
The  key difference compared to the case $\obspace = \reals$ is due to the possible discontinuity of the conditional probability densities $\oprob_{iy}(\action)$ at the end points $\obsl_\action$ and $\obsu_\action$.
Without  appropriate assumptions,
$\psi(\lambda)$ defined  in
(\ref{eq:lambdafn}) can become negative in two ways:
(i) If  $\obspace^\lambda_{2} = \emptyset$ and $ \obspace^\lambda_1$
is non-empty
(ii)  $\bar{\obspace}^\lambda_{2} = \emptyset$ and $\bar{ \obspace}^\lambda_1 $
is non-empty.
Assumption  \ref{obsinit}, see  Theorem \ref{thm:boundary}, ensures that  these two cases do not occur.

To prove $\psi(\lambda) \geq 0$ for $\lambda \in [0,1]$, boundary conditions need to be handled. Define $\lami,\lamii,\lamiii,\lamiv$ as
%\in [0^-,1^+]$ as
  \begin{equation} \begin{split}
     \lami &= \sup\{\lambda: \bar{\obspace}_1^\lambda = \emptyset,
     \bar{\obspace}_{2}^\lambda = \emptyset\} \\
\lamii &= \sup\{\lambda: \bar{\obspace}_1^\lambda = \emptyset, \bar{\obspace}_{2}^\lambda \neq \emptyset\}; \\
\lamiii &= \inf\{\lambda: {\obspace}_1^\lambda \neq \emptyset, {\obspace}_{2}^\lambda = \emptyset\}; \\
\lamiv &= \inf\{\lambda: {\obspace}_1^\lambda = \emptyset, {\obspace}_{2}^\lambda = \emptyset\}.
\end{split}
\end{equation}
Clearly,  $\lami \leq \lamii \leq \lamiii \leq \lamiv $ 
since $\level^\p \filter(\belief,\obs,\action)$ is increasing in $\obs$ by \ref{TP2_obs} and $\obspace_u^\lambda \subseteq \obspace_{u}^{\bar{\lambda}}$ for $\lambda < \bar{\lambda}$.  We now consider $\lambda\in [0,1]$ split
into the following 5 sub-cases and show that $\psi(\lambda) \geq 0$ for each sub-case:
\begin{equation}
 \begin{split}
\text{ Case 1a. } &
\lambda \in [0,\lami] \iff  \bar{\obspace}_1^\lambda = \emptyset, \bar{\obspace}_{2}^\lambda = \emptyset \\
\text{ Case 1b. } &
\lambda \in (\lami,\lamii] \iff  \bar{\obspace}_1^\lambda = \emptyset, \bar{\obspace}_{2}^\lambda \neq \emptyset \\
\text{ Case 1c. } &
\lambda \in (\lamii,\lamiii] \iff  \bar{\obspace}_1^\lambda \neq
\emptyset, \bar{\obspace}_{2}^\lambda \neq \emptyset \\
\text{ Case 1d. } &
\lambda \in (\lamiii,\lamiv] \iff  {\obspace}_1^\lambda = \emptyset, {\obspace}_{2}^\lambda \neq \emptyset  \\
   \text{ Case 1e. }  &
   \lambda \in (\lamiv,1] \iff  {\obspace}_1^\lambda = \emptyset, {\obspace}_{2}^\lambda =\emptyset.
 \end{split}
 \label{eq:subcases}
\end{equation}
Note that
 (\ref{eq:firstline}) implies $\psi(\lambda) = 0$ for Case 1a and Case 1e.
For Case 1b, re-expressing $\bar{\obspace}_{2}^\lambda = \{\obs:  -(\level - \lambda \ones)^\p
\oprob_{y}({2})  \tp^\p \belief > 0 \} $,
(\ref{eq:lastline})
implies that $\psi(\lambda) \geq 0$.
Equivalently, (\ref{eq:gradientexp}) implies $d\psi(\lambda)/d\lambda > 0$;   since $\psi(\lami) = 0$, therefore
$\psi(\lambda) \geq 0$ for $\lambda \in (\lami,\lamii]$.
%Before proceeding with Cases 2b and 2d  note that the use of stationary points to prove $\psi(\lambda) \geq 0$
%(which we used in Case 1) is not useful. 
%Indeed, it is clear from (\ref{eq:gradientexp}) that for Cases 2b and 2d, $d\psi/d\lambda\neq 0 $.
%So we  establish that $\psi(\lambda) \geq 0$ in Cases 2b and 2d  by construction  - it is here that \ref{obsinit} plays a crucial role.
For Case 1d,
it follows immediately from (\ref{eq:lambdafn}) that $\psi(\lambda) \geq 0$.
Equivalently, (\ref{eq:gradientexp}) implies $d\psi(\lambda)/d\lambda < 0$;
since $\psi(\lamiv)=0$, therefore $\psi(\lambda) \geq 0$ for
$\lambda \in (\lamiii,\lamiv)$.\\
Finally, for Case~1c, since both $\obspace_1^\lambda$ and $\obspace_{2}^\lambda$ are non-empty, the single crossing
condition \ref{obssc} kicks in and an identical 
argument as the case $\obspace = \reals$  applies.
Indeed,  $\psi(\lamii) \geq 0$, $\psi(\lamiii) \geq 0$, and $\psi(\lambda)$ is differentiable
for $\lambda \in  (\lamii,\lamiii)$; so  %an identical
% argument to Case~1 applies, namely
$\psi(\lambda) \geq 0$ for
$\lambda \in (\lamii, \lamiii)$ because 
$\psi(\lambda^*) \geq 0$ at each stationary point $\lambda^* \in (\lamii,\lamiii)$.

{\em Remark}:
The case $\obspace = \reals$ (Theorem  \ref{thm:convexdom})  can be viewed as a special instance of (\ref{eq:subcases})
with $\lami =\lamii= 0$, and $\lamiii = \lamiv = 1$ (but to enhance clarity we described it before Case 1). The main point when $\obspace = \reals$  is that $\obspace^\lambda_1, \obspace^\lambda_{2} 
$ are never empty for $\lambda \in (0,1)$ and therefore only  Case 1c occurs.

{\bf Case 2. $\obspace_\action$ is finite}:  Finally, we prove~(\ref{eq:lambdafn}) for  the case 
  $\obspace_\action = \{1,2,\ldots,\obsdim_\action\}$.
  Construct  the piecewise constant probability density function
  $O_{io}(\action) = \oprob_{iy}(\action)$ for $o \in [y, y+1)$ and $y \in \{1,2,\ldots,\obsdim_\action\}$.
It is easily verified that  $\filter(\belief,o,\action)
= \filter(\belief,\obs,\action)$, $\filterd(\belief,o,\action) =\filterd(\belief,\obs,\action)$,  and the value function and optimal policy remain unchanged. Then the above proof for Case 1 (finite support) applies. \hfill \qed

{\em Remark}. To emphasize the importance of sufficient condition \ref{obsinit},
the following examples show that \ref{obsinit} is in some sense  necessary;
when %\ref{obsinit}
it fails to hold, then $\psi(\lambda) < 0$ for some interval of
$\lambda$ and convex dominance does not hold.
%To show the usefulness of sufficient condition \ref{obsinit}, here are  examples showing that when %\ref{obsinit} fails to hold then $\psi(\lambda) < 0$ for some interval of
%$\lambda$.
Consider
 $\statedim = 3, \obsdim = 3$, $\belief = \begin{bmatrix} 0.2 &  0.3 & 0.5 \end{bmatrix}^\p $, $\level=[0,0,1]^\p$.\\
Example 1.   
      $\tp = \begin{bmatrix} 0.9 & 0.1 & 0.1 \\ 0.1 & 0.8 & 0.1 \\ 0 & 0.1 & 0.9 
   \end{bmatrix}$, \\ $\oprob(1)=
   \begin{bmatrix} 0.7 &0.2 &0.1 \\  0.1& 0.3 &0.6 \\ 0& 0.1 &0.9 \end{bmatrix}$,
   $\oprob(2) = \begin{bmatrix} 0.8 & 0.1 & 0.1 \\ 0.2 & 0.2 & 0.6\\ 0.05 & 0.05 & 0.9 \end{bmatrix}$.\\
   Then  $\phi(\lambda) < 0 $ for $\lambda \in (0,0.26]$.
   This example violates Case 1b.
\\
   Example 2.      $\tp = \begin{bmatrix} 0.9 & 0.1 & 0.1 \\ 0.1 & 0.8 & 0.1 \\ 0 & 0.1 & 0.9 
   \end{bmatrix}$,  \\
$\oprob(1) = \begin{bmatrix} 0.8 &  0.2 &  0 \\  0.1 &  0.8 &  0.1 \\  0 & 0.2 & 0.8
\end{bmatrix}
$,
$\oprob(2)=\begin{bmatrix} 0.8 & 0.1 & 0.1\\  0.1 & 0.3 & 0.6\\ 0 & 0.1 & 0.9
\end{bmatrix}
$.\\  Then $\psi(\lambda)< 0$ for $\lambda \in [0.25,0.93]$. This example violates Case 1d.

\subsection{Proof of Theorem \ref{thm:main} when $\statespace = \reals$ and $\obspace = \reals$} \label{sec:cont}
In complete analogy to (\ref{eq:lambdafn}) convex dominance holds if for $\lambda \in \reals$,
\begin{multline}   \psi(\lambda) \ole \int_{\obs \in \obspace^k_{2}} [
\cm_2(\obs_{1:k}, \belief_0) - 
 \lambda]^+ \filterd(\belief,\obs_{1:k},2) \\
 - \int_{\obs \in \obspace^k_1}[ \cm_1(\obs_{1:k}, \belief_0)- \lambda]^+ \filterd(\belief,\obs_{1:k},1) \geq 0. \label{eq:lambdafn2}
\end{multline}
where 
 $$ \cm_\action(\obs_{1:k}, \belief_0) = \langle \state, \filter(\belief_0, \obs_{1:k},\action) \rangle  \ole \int_\statespace \state \,\filter(\belief_0, \obs_{1:k},\action)(x)\, dx$$

 Defining $   \obspace^{k,\lambda}_u = \{\obs_{1:k}: \langle \state, \filter(\belief,\obs_{1:k},
 \action )\rangle > \lambda\}$, \begin{align*}
\psi(\lambda)    & = \int_{\obspace^{k,\lambda}_{2}}  [ \langle \state, \filter(\belief,\obs_{1:k},{2})\rangle - \lambda]\, \filterd(\belief,\obs_{1:k},{2})  \\ 
               & \qquad -     \int_{\obspace^{k,\lambda}_1} [ \langle \state,  \filter(\belief,\obs_{1:k},1) \rangle - \lambda] \, \filterd(\belief,\obs_{1:k},1)
                   % \label{eq:firstline2}
                                  \\
  %                &= \int_{ \obspace^\lambda_{2}} [\level^\p \oprob_{\obs_k}({2}) \tp^\p \cdots
% \oprob_{\obs_1}(2) \tp^\p \, 
%                    \belief - \lambda \ones^\p  \oprob_\obs({2}) \tp^\p \belief]
% - \sum_{\obs \in \obspace^\lambda_\action} [ \level^\p \oprob_\obs(\action)\tp^\p \belief- \lambda \ones^\p  \oprob_\obs(\action) \tp^\p \belief]   \nonumber
% \\
% %&= (\level-\lambda \ones)^\p \left[ \sum_{\obs  \in \obspace^\lambda_{\action+1}}
% %     \oprob_\obs({\action+1}) - \sum_{\obs  \in \obspace^\lambda_\action} %\oprob_\obs(\action) \right] \tp^\p \,\belief
% %\nonumber
% %     \\
% &= \left\langle (\state - \lambda \ones),  \left[ \int_{{\obspace}^{k,\lambda}_2}
%            \oprob_{\obs_k}(2)  \cdots  \oprob_{\obs_1}(2) 
%            - \int_{ {\obspace}^{k,\lambda}_{1}} \oprob_{\obs_k}({1}) 
% \cdots  \oprob_{\obs_1}({1}) 
%              \right]  \,\belief \right\rangle \nonumber\\
                 &=
                   \bigl\langle  (x -\lambda \ones),  \bigl[ \bcdf_2(z_k) \bcdf_2(z_{k-1})  \cdots \bcdf_2(z_1)  \\  & \qquad \qquad \qquad -   \bcdf_1(\bz_k)  \bcdf_1(\bz_{k-1})  \cdots \bcdf_1(\bz_1)
                   \bigr]\,\belief \bigr \rangle  %\label{eq:lastline2} 
\end{align*}
for some $z_1,\ldots,z_k \in \reals$ and $\bz_1,\ldots,\bz_k \in \reals$ which depend on $\lambda$.

In complete analogy to Lemma \ref{lem:nabla}, 
$\psi(\lambda) = 0 $ for $\lambda \rightarrow -\infty$  and $\lambda = \infty$
and  $\psi(\lambda) $ is continuously differentiable  wrt $\lambda \in \reals$ with gradient
  \begin{multline}  \frac{d\psi(\lambda)}{d\lambda} = -
\bigl\langle  \ones,  \bigl[ \bcdf_2(z_k) \bcdf_2(z_{k-1})  \cdots \bcdf_2(z_1) \\ -   \bcdf_1(\bz_k)  \bcdf_1(\bz_{k-1})  \cdots \bcdf_1(\bz_1) 
                   \bigr]\,\belief \bigr \rangle    \label{eq:gradientexp2} 
               \end{multline}
               The remainder of the proof is similar to that of Theorem \ref{thm:convexdom}.

\bibliographystyle{IEEEtran}
 %\bibliography{C:/Users/vikramk/styles/bib/vkm}

\begin{IEEEbiographynophoto}{Vikram Krishnamurthy}
(F'05) received the Ph.D. degree from the Australian
National University
in 1992. He is currently a professor in the
School of Electrical \& Computer Engineering,
Cornell University. From 2002-2016 he
was a Professor and Canada Research Chair
at the University of British Columbia, Canada.
His research interests include statistical signal
processing  and
stochastic control in social networks and adaptive sensing. He served
as Distinguished Lecturer for the IEEE Signal Processing Society and
Editor-in-Chief of the IEEE Journal on Selected Topics in Signal Processing.
In 2013, he was awarded an Honorary Doctorate from KTH
(Royal Institute of Technology), Sweden. He is author of the books
{\em Partially Observed Markov Decision Processes} and
{\em Dynamics of Engineered Artificial Membranes and Biosensors} published by Cambridge
University Press in 2016 and 2018, respectively.
\end{IEEEbiographynophoto}

\end{document}